\documentclass[10pt]{article}
\usepackage[affil-it]{authblk}
\usepackage{geometry}
\usepackage{amsmath}

\usepackage{amsthm}
\usepackage{amssymb}
\usepackage{amsfonts}
\usepackage{dsfont}
\usepackage{accents}
\usepackage{tabularx} 
\usepackage{caption}
\usepackage{multirow}
\usepackage{booktabs}
\usepackage{graphicx}
\graphicspath{{Figures/}}
\usepackage{subfigure}
\usepackage{tikz,pgfplots}
\usepgfplotslibrary{groupplots}
\usetikzlibrary{intersections,calc,arrows,matrix,spy}
\usepackage{color}
\definecolor{darkred}{rgb}{0.6,0,0}
\definecolor{darkgreen}{rgb}{0,0.5,0}
\definecolor{darkblue}{rgb}{0,0,0.5}
\definecolor{SkyBlue}{rgb}{0.53, 0.81, 0.92}
\pgfplotsset{compat=1.5.1}
\newcommand\irregularcircle[2]{
  \pgfextra {\pgfmathsetmacro\len{(#1)+rand*(#2)}}
  +(0:\len pt)
  \foreach \a in {10,25,...,350}{
    \pgfextra {\pgfmathsetmacro\len{(#1)+rand*(#2)}}
    -- +(\a:\len pt)
  } -- cycle
}
\usepackage{algorithm}
\usepackage{algorithmic}

\usepackage[squaren,Gray]{SIunits}

\newcommand{\argmin}[2]{\mathrm{arg}\,\underset{#1}{  \mathrm{min}} \; #2}  
\def\R{\mathbb{R}}          									 	          
\def\N{\mathbb{N}}          									    	      
\def\C{\mathbb{C}}          									        	  
\def\nab{\mbox{{\boldmath{$\nabla$}}}}
\renewcommand\Re{\mathrm{Re}}          									 	          
\def\e{\mathrm{e}}          									        	  
\def\ii{\mathrm{i}}          									        	  
\newcommand{\fd}{\mathrm{\mathbf{f}}}
\newcommand{\qd}{\mathrm{\mathbf{q}}}
\newcommand{\yd}{\mathrm{\mathbf{y}}}
\newcommand{\ud}{\mathrm{\mathbf{u}}}
\newcommand{\vd}{\mathrm{\mathbf{v}}}
\newcommand{\wdd}{\mathrm{\mathbf{w}}}
\newcommand{\Gd}{\mathrm{\mathbf{G}}}

\newcommand{\Jd}{\mathrm{\mathbf{J}}}
\newcommand{\Id}{\mathrm{\mathbf{I}}}
\newcommand{\kd}{\mathrm{\mathbf{k}}}
\newcommand{\xd}{\mathrm{\mathbf{x}}}
\newcommand{\dd}{\mathrm{d}}
\newcommand{\dr}{\partial}
\newcommand{\diag}[1]{\mathrm{\mathbf{diag}}(#1)}
\newcommand{\ie}{\textit{i.e., }}                                            
\newcommand{\eg}{\textit{e.g., }}                                            

\newtheorem{lemma}{Lemma}[section]

\newtheorem{proposition}[lemma]{Proposition}

\begin{document}
\title{Efficient Inversion of Multiple-Scattering Model for Optical Diffraction Tomography}

\author[1,$\dagger$]{Emmanuel Soubies}
\author[1,$\dagger$]{Thanh-An Pham}
\author[1,*]{Michael Unser}
\affil[1]{Biomedical Imaging Group, \'Ecole polytechnique f{\'e}d{\'e}rale de Lausanne
(EPFL), CH-1015 Lausanne, Switzerland.}
\date{}
 \maketitle
 
\renewcommand{\thefootnote}{}
\footnote{\noindent$\dagger$ These authors contributed equally to this paper. Emails: \{emmanuel.soubies,thanh-an.pham,michael.unser\}@epfl.ch}
\renewcommand{\thefootnote}{\arabic{footnote}}

\begin{abstract}
Optical diffraction tomography  relies on solving an inverse scattering problem governed by the wave equation. Classical reconstruction algorithms are based on linear approximations of the forward model (Born or Rytov), which limits their applicability to thin samples with low refractive-index contrasts. More recent works have shown the  benefit of adopting nonlinear models. They account for multiple scattering and reflections, improving the quality of reconstruction. To reduce the complexity and memory requirements of these methods, we derive an explicit formula for the Jacobian matrix of the nonlinear Lippmann-Schwinger model which lends itself to an efficient evaluation of the gradient of the data-fidelity term. This allows us to deploy efficient methods to solve the corresponding inverse problem subject to sparsity constraints.
\end{abstract}


\section{Introduction}  

	Optical diffraction tomography  (ODT) was introduced in~\cite{Wolf1969} by E. Wolf in the late '60s. It is a microscopic technique that retrieves the distribution of refractive indices in biological samples out of holographic measurements of the scattered complex field produced when the sample is illuminated by an incident wave.  This method is of particular interest in biology because, contrarily to fluorescence imaging, it does not require any staining of the sample~\cite{Jin2017}. It proceeds by solving an inverse scattering problem, where the scattering phenomenon is governed by the wave equation. There is a vast literature on inversion methods going from linearized models (Born, Rytov) \cite{Wolf1969,Devaney1981} to nonlinear ones~\cite{Mudry2012,Kamilov2016,Liu2016,Liu2017}. It is worth noting that the scattering model, along with its associated inverse problem, is generic and not limited to optical diffraction tomography. In particular, it is encountered in many other fields such as acoustics, microwave imaging, or radar applications~\cite{Colton2012}.

\subsection{From the wave equation to the Lippmann-Schwinger integral equation}

	Let us consider an unknown object of refractive index $n(\xd)$ lying in the region $\Omega \subseteq \R^D$ ($D \in \{2,3\}$) and being immersed in a medium of refractive index $n_{\mathrm{b}}$, as depicted in Fig.~\ref{fig:schema}. This sample is  illuminated by the incident plane wave
	\begin{equation}\label{eq:PlaneWave}
		u^{\mathrm{in}}(\xd,t) = \Re \left( u_0 \e^{\ii \kd \cdot  \xd - \ii \omega t }   \right) , 
	\end{equation}
where the wave vector  $\kd \in \R^D$  specifies the direction of the wave propagation, $\omega \in \R$ denotes its angular frequency, and  $u_0 \in \C$ defines its complex envelope (amplitude). The resulting total electric field $u(\xd,t)$ satisfies the wave equation
\begin{equation}\label{ea:waveEq}
	\nabla^2 u(\xd,t) - \frac{n^2(\xd)}{\mathrm{c}^2} \frac{\dr^2 u}{\dr t^2} (\xd,t) =0,
\end{equation}
where $\mathrm{c}\simeq 3\times 10^8$\meter\per\second~is the speed of light in free space. Denoting by $u(\xd)$ the complex amplitude of  $u(\xd,t)= \Re \left(  u(\xd)  \e^{ -\ii \omega t }   \right) $ and substituting it into~\eqref{ea:waveEq}, we obtain the  inhomogeneous Helmholtz equation
\begin{equation}\label{eq:Helmholtz1}
	\nabla^2 u(\xd) + k_0^2 n^2(\xd) u(\xd) =0, 
\end{equation}
with the propagating constant in free space $k_0 = \omega / \mathrm{c}$. The total field $u(\xd)$ is the sum of the scattered field $u^{\mathrm{sc}}(\xd)$ and of the incident field $u^{\mathrm{in}}(\xd)$, which  is itself a solution of the homogeneous Helmholtz equation $\nabla u^{\mathrm{in}}(\xd) + k_0^2 n_{\mathrm{b}}^2 u^{\mathrm{in}}(\xd) =0$. Accordingly, \eqref{eq:Helmholtz1} can be rewritten as (see \cite{Wolf1969})
\begin{equation}
	\nabla^2 u^{\mathrm{sc}}(\xd) + k_0^2 n_{\mathrm{b}}^2  u^{\mathrm{sc}}(\xd) =-f(\xd)u(\xd), 
\end{equation}
where $f(\xd) = k_0^2 (n^2(\xd)-n_{\mathrm{b}}^2)$ defines the   scattering potential function.  It follows that
\begin{equation}
	u^{\mathrm{sc}}(\xd) =   \int_{\Omega}  g(\xd-\xd') f(\xd')u(\xd')\,  \dd \xd',
\end{equation}
where $g(\xd)$ is the Green's function of the shift-invariant differential operator ($\nabla^2 + k_0^2 n_{\mathrm{b}}^2 \Id$).  Specifically, $g$  verifies  $\nabla^2 g(\xd)  + k_0^2 n_{\mathrm{b}}^2 g(\xd)  = -\delta(\xd) $, where $\delta$ is the Dirac distribution and the minus sign is a convention used in physics. Under Sommerfeld's radiation condition, $g(\xd)$ is given by \cite[and references therein]{Schmalz2010}
\begin{equation}
	g(\xd) = \left\lbrace 
	\begin{array}{ll}
		\frac{1}{4}H_0^{(1)} (k_0 n_{\mathrm{b}}\|\xd\|), & D=2, \\
		\frac{1}{4\pi} \frac{\e^{\ii k_0 n_{\mathrm{b}}\|\xd\|}}{\|\xd\|}, & D=3.
	\end{array}\right. 
\end{equation}
There, $H_0^{(1)}$ is the Hankel function of the first kind.
Finally, the total field $u(\xd)$ is governed by the Lippmann-Schwinger equation
\begin{equation}\label{eq:LippScwing}
	u(\xd) = u^{\mathrm{in}}(\xd) + \int_{\Omega}  g(\xd-\xd') f(\xd')u(\xd')\,  \dd \xd'.
\end{equation}

\begin{figure}
	\centering
	\resizebox {0.9\textwidth} {!} {
	\begin{tikzpicture}
		\draw[dotted,thick]  (-2.2,-2.2) rectangle (2.2,2.2);
		\draw[gray,fill=gray]  (5.5,-3) rectangle (5.7,3);
		\draw[gray,dashed,thin]  (-5.6,-3) -- (-5.6,3);
		 \draw[rounded corners=1mm,fill,inner color=SkyBlue!90,outer color=SkyBlue!20] (0,0) \irregularcircle{1.1cm}{0.2cm};
		  \draw[fill,inner color=brown!90,outer color=brown!60] (0.3,0.2)ellipse (0.2 and 0.3);
		  \node at (-0.2,-0.2) {$n(\xd)$};
		  \node at (-1.5,-1.5) {$n_{\mathrm{b}}$};
		\draw[blue,thick] plot[smooth] coordinates {(3.3,2.5) (4.3,1.8) (4,1.3) (4.5,0.6)  (3.7,-0.1) (4,-0.7) (3.5,-1.5) (3.8,-2.1) (3.2,-2.5)};
		\draw[red,thick] plot[smooth] coordinates {(-3.2,2.5) (-3.8,2) (-3.8,1.5) (-4.2,1)  (-3.5,-0.1) (-4.2,-0.7) (-3.5,-1.5) (-3.6,-2.1) (-3.4,-2.5)};
		\node[align=center] at (3.5,3) {Forward scattered \\ wave $u^{\mathrm{sc}}$};
		\node[align=center] at (-3.5,3) {Backward scattered \\ wave $u^{\mathrm{sc}}$};
		\node[align=center] at (0,2.8) {Sample};
		\node[align=center] at (6,2.7) {$\Gamma$};
		\draw[-latex,thick,opacity=0.2] (-5.6,1.5)  -- ++ (-20:1cm);
		\draw[-latex,thick]  (-5.6,0) --  ++(0:1cm);
		\node at (-5,0.2) {$\kd$};
		\draw[-latex,thick,opacity=0.2]  (-5.6,-1.5) -- ++(20:1cm);
		\draw  (-5.2,0) --  ++(90:1.2cm);\draw  (-5.2,0) --  ++(-90:1.2cm);
		\draw  (-5.4,0) --  ++(90:1.2cm);\draw  (-5.4,0) --  ++(-90:1.2cm);
		\draw  (-5.6,0) --  ++(90:1.2cm);\draw  (-5.6,0) --  ++(-90:1.2cm);
		\draw  (-5.8,0) --  ++(90:1.2cm);\draw  (-5.8,0) --  ++(-90:1.2cm);
		\draw  (-6,0) --  ++(90:1.2cm);\draw  (-6,0) --  ++(-90:1.2cm);
		\draw[opacity=0.2]   (-5.6,1.5)   ++ (160:0.4cm) --  ++(70:1.2cm);\draw[opacity=0.2]   (-5.6,1.5)   ++ (160:0.4cm)--  ++(-110:1.2cm);
		\draw[opacity=0.2]  (-5.6,1.5)   ++ (160:0.2cm)  --  ++(70:1.2cm);\draw[opacity=0.2]  (-5.6,1.5)   ++ (160:0.2cm) --  ++(-110:1.2cm);
		\draw[opacity=0.2]  (-5.6,1.5)   ++ (-20:0cm)  --  ++(70:1.2cm);\draw[opacity=0.2]  (-5.6,1.5)   ++ (-20:0cm) --  ++(-110:1.2cm);
		\draw[opacity=0.2]  (-5.6,1.5)   ++ (-20:0.2cm)  --  ++(70:1.2cm);\draw[opacity=0.2]   (-5.6,1.5)   ++ (-20:0.2cm)  --  ++(-110:1.2cm);
		\draw[opacity=0.2]  (-5.6,1.5)   ++ (-20:0.4cm)  --  ++(70:1.2cm);\draw[opacity=0.2]  (-5.6,1.5)   ++ (-20:0.4cm) --  ++(-110:1.2cm);
		\draw[opacity=0.2]   (-5.6,-1.5)   ++ (200:0.4cm)--  ++(110:1.2cm);\draw[opacity=0.2]   (-5.6,-1.5)   ++ (200:0.4cm) --  ++(-70:1.2cm);
		\draw[opacity=0.2]  (-5.6,-1.5)   ++ (200:0.2cm)--  ++(110:1.2cm);\draw[opacity=0.2]  (-5.6,-1.5)   ++ (200:0.2cm) --  ++(-70:1.2cm);
		\draw[opacity=0.2]  (-5.6,-1.5)   ++ (20:0cm) --  ++(110:1.2cm);\draw[opacity=0.2]  (-5.6,-1.5)   ++ (20:0cm) --  ++(-70:1.2cm);
		\draw[opacity=0.2]  (-5.6,-1.5)   ++ (20:0.2cm) --  ++(110:1.2cm);\draw[opacity=0.2]  (-5.6,-1.5)   ++ (20:0.2cm) --  ++(-70:1.2cm);
		\draw[opacity=0.2]  (-5.6,-1.5)   ++ (20:0.4cm) --  ++(110:1.2cm);\draw[opacity=0.2]  (-5.6,-1.5)   ++ (20:0.4cm) --  ++(-70:1.2cm);
		\node[rotate=90] at (6,0) {Detector plane ($\yd_p$)};
		\node[rotate=90] at (-6.7,0) {Sources ($u_p^{\mathrm{in}}$)};
		\node at (1.9,1.9) {$\Omega$};
	\end{tikzpicture}
	}
	\caption{\label{fig:schema} Optical diffraction tomography. A sample of refractive index $n(\xd)$ is immersed in a medium of index $n_b$ and illuminated by an incident plane wave (wave vector $\kd$). The interaction of the wave with the object produces  forward and backward scattered waves. The forward scattered wave is recorded in the detector plane. Optionally,  a second detector plane may record the backward scattered wave (see Section~\ref{sec:ExpNum}). }
\end{figure}
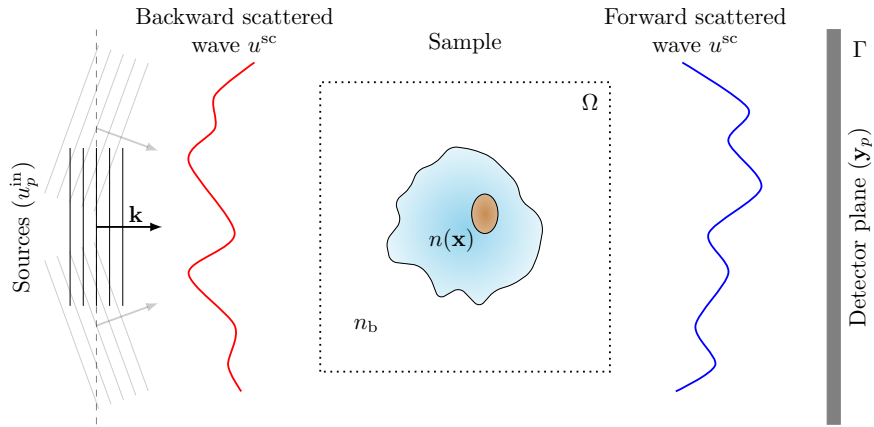

\subsection{Inverse ODT problem: prior work}   Let the object be illuminated by a series of incident fields $u_p^{\mathrm{in}}(\xd)$, $p \in [1\ldots P]$. Records of the resulting total fields $u_p(\xd)$ at  positions $\xd_m$ ($m \in [1\ldots M]$) in the detector plane $\Gamma$  are denoted $\yd_p \in \C^M$ (see Fig.~\ref{fig:schema}). The objective is then to retrieve the scattering potential function  $f(\xd)$ (\ie the refractive index $n(\xd)$) from the data $\yd_p$. Pioneering methods were using linear approximations of the model. For instance, assuming that the scattering field is weak compared to the incident one (\ie $u^{\mathrm{sc}} \ll u^{\mathrm{in}}$), one can interpret the phase of the transmitted wave as the Radon transform of the refractive index and then reconstruct $f$ using the filtered-back-projection algorithm~\cite{Kak2001,Choi2007}. This method ignores the effect of diffraction. The first Born approximation~\cite{Wolf1969} has then been proposed as a refined model. Its validity is however limited to thin samples with weak variations of their refractive index (RI)~\cite{Chen1998}. A more accurate linearization, less sensitive to the thickness of the sample but still limited to weak RI contrasts, is given by the Rytov approximation~\cite{Devaney1981,Sung2009}. It is derived by assuming that the total field has the form $u(\xd) = u^{\mathrm{in}}(\xd)\e^{\phi(\xd)}$, where $\phi(\xd)$ is a complex phase function.  Both Born and Rytov approximations have been originally used to derive direct inversion methods. They were later used within regularized variational approaches to improve the quality of reconstructed images~\cite{Sung2011,Lim2015}.

	Inversion methods that use a nonlinear model have been shown to significantly improve the accuracy of reconstruction. These include the conjugate-gradient method (CGM)~\cite{Chaumet2009,Belkebir2005}, the contrast source-inversion method (CSI)~\cite{Abubakar2002}, the beam-propagation method (BPM)~\cite{Kamilov2015}, the recursive Born approximation~\cite{Kamilov2016}, or the hybrid method proposed in~\cite{Mudry2012}.  Although still approximate (for instance, they do not properly take reflections into account), they more closely adhere to the model of the physical phenomenon than the linear models, at the price of a higher computational cost.  We refer the reader to  \cite{Jin2017} for additional details concerning existing approximations, regularizations, algorithms, and comparisons.

	To address applications with thick samples and large RI contrasts, a better solution is to rely on the exact Lippmann-Schwinger model which accounts for mutiple scattering and reflections. Such an approach has been recently proposed in~\cite{Liu2016,Liu2017} (SEAGLE algorithm). There, the authors tackle the problem from a variational perspective. They minimize a nonconvex objective using the well known fast iterative shrinkage-thresholding algorithm (FISTA)~\cite{Beck2009}. Their main contribution is to compute the forward model (which itself requires the inversion of an operator) using  Nesterov's accelerated gradient-descent (NAGD) method \cite{Nesterov1983} and, more interestingly, to explicitly compute the gradient of the quadratic data-fidelity term as an error-backpropagation of the forward algorithm.  However, the bottleneck of their method is its high memory consumption. Indeed, the error-backpropagation strategy requires one to store all the iterates produced during the computation of the iterative forward model. This can be limiting for large 3D volumes. 
	

\subsection{Contributions} 

 To improve the computational efficiency of solvers such as SEAGLE, we provide an explicit expression for the Jacobian of the nonlinear Lippmann-Schwinger operator. This results in an efficient method to compute the gradient of the data-fidelity term and avoid recoursing to the memory-consuming error-backpropagation strategy. Another advantage is that the computation of the forward model and of the gradient are now decoupled. They can thus be solved using any numerical scheme. Then, considering simulated data, we show that the proposed method results in a significant reduction of both computational time and  memory requirements with respect to SEAGLE, at no loss in  quality. 
 
 	In Section~\ref{sec:Forward}, we formulate the discrete forward model proposed in~\cite{Liu2017}. Then, the common approach used to solve the inverse problem subject to sparsity constraints is presented in Section~\ref{sec:OptiPb}. There, we highlight our main innovation with respect to SEAGLE, which is a new computation of the gradient of the data-fidelity term. It relies on the derivation of the Jacobian of the forward model, which is given by Proposition~\ref{propo}. Finally, Sections~\ref{sec:AlgoAnalysis} and~\ref{sec:ExpNum} are dedicated to numerical comparisons.

\section{Solving the inverse problem}

\subsection{Formulation of the forward model} \label{sec:Forward}
	
	In this section, we review the formulation of the forward model that was  proposed by Liu \textit{et al.} in \cite{Liu2017}. Let the region of interest $\Omega$ be divided into $N\in \N$ ``pixels''. Then, over $\Omega$, we define the discrete version of~\eqref{eq:LippScwing} as
\begin{equation} \label{eq:LippScwingDiscret}
	\ud_p = \ud^{\mathrm{in}}_p + \Gd \, \diag{\fd} \ud_p,
\end{equation} 
where $\ud_p \in \C^N$, $\ud^{\mathrm{in}}_p \in \C^N$, $\fd \in \R^N$ are the discrete representations of  $u_p$, $u^{\mathrm{in}}_p$, and $f$, respectively. The diagonal matrix $ \diag{\fd}\in  \R^{N \times N}$ is formed out of  the entries of $\fd$, while $\Gd \in \C^{N \times N}$ stands for the matrix of the convolution operator on $\Omega$ (convolution with $g$). One can notice that \eqref{eq:LippScwingDiscret} is nonlinear with respect to $\fd$. On the other hand, given $\ud^{\mathrm{in}}_p$ and $\fd$, the computation of $\ud_p$ amounts to inverting the operator ($\Id - \Gd \, \diag{\fd}$). Instead of attempting to compute this inverse directly, the ODT forward model on $\Omega$, for a given $\fd$,  is defined as
\begin{equation}\label{eq:ForwardOptPb}
	\ud_p(\fd) = \argmin{\ud \in \C^N}{\frac12 \| (\Id - \Gd \,\diag{\fd})  \ud -  \ud^{\mathrm{in}}_p  \|_2^2} .
\end{equation}
This classical quadratic-minimization problem can be solved iteratively using numerous state-of-the-art algorithms (see Section~\ref{sec:ForwardComputation}).
Then, from the total field $\ud_p(\fd)$ (inside $\Omega$), we get measurements $\yd_p$ on $\Gamma$ using a different discretization $\tilde{\Gd} \in \C^{M \times N}$  of the Green's function (see \cite{Liu2017})
\begin{equation}
	\yd_p = \tilde{\Gd} \, \diag{\fd} \ud_p(\fd) +\ud_p^{\mathrm{in}}\vert_{\Gamma},
\end{equation}
where $\ud_p^{\mathrm{in}}\vert_{\Gamma}$ denotes the restriction of the field $\ud_p^{\mathrm{in}}$ to the area~$\Gamma$.

\subsection{Common optimization strategy}\label{sec:OptiPb}
Following the classical variational approach, the estimation of $\fd \in \R^N$  from the measurements $\{\yd_p \in \C^{M}\}_{p\in [1\ldots P]}$ is formulated as the  optimization problem
\begin{equation}\label{eq:OptiPb}
	\widehat{\fd} \in \left\lbrace \argmin{\fd \in \R^N}{ \left( \mathcal{D}(\fd) + \mu \mathcal{R}(\fd)}\right)  \right\rbrace, 
\end{equation}
where $\mathcal{D} : \R^N  \rightarrow \R$ measures the fidelity to data, $\mathcal{R}   : \R^N \rightarrow \R$ imposes some prior to the solution (regularization), and $\mu >0$ balances between these two terms. It is customary to consider the  data term
\begin{equation}\label{eq:DataTerm}
	\mathcal{D}(\fd) = \sum_{p=1}^P \mathcal{D}_p(\fd),
\end{equation}
where $\forall p \in [1\ldots P]$
\begin{equation}
	\mathcal{D}_p(\fd) = \frac12 \|\tilde{\Gd} \, \diag{\fd}\ud_p(\fd) - \yd_p^{\mathrm{sc}}\|^2_2,
\end{equation}
which is well suited for Gaussian noise.
Here, $\yd_p^{\mathrm{sc}} = (\yd_p - \ud_p^{\mathrm{in}}\vert_{\Gamma})$ is the scattered measured field at the detector plane $\Gamma$ and  $\ud_p(\fd)$ is given by~\eqref{eq:ForwardOptPb}. As regularizer $\mathcal{R}$, the combination 
\begin{equation}\label{eq:Regul}
	\mathcal{R}(\fd) = i_{\geqslant 0}(\fd) + \| \nab \fd \|_{2,1}  = i_{\geqslant 0}(\fd) + \sum_{n=1}^N \sqrt{\sum_{d=1}^D (\dr_d \fd)_n^2}
\end{equation}
of total variation (TV) penalty and nonnegativity constraint is used, where $i_{\geqslant 0}(\fd) = \{ 0, \text{ if } \fd_n \geq 0 \, \forall n ; +\infty, \text{ otherwise} \}$ and $\dr_d$ denotes the gradient operator along the $d$th direction. This choice is supported by the facts that we consider situations where $n_{\mathrm{b}} \leq n(\xd) \, \Rightarrow f(\xd) \geq 0$ and that $n$ and, thus, $f$ can be assumed to be piecewise-constant. It is worth noting that \eqref{eq:OptiPb} is nonconvex due to the nonlinearity of the forward operator in~\eqref{eq:LippScwingDiscret}. However, since $\mathcal{D}$ is smooth with respect to $\fd$,~\eqref{eq:OptiPb} can be solved by deploying a forward-backward splitting (FBS) method~\cite{Combettes2005} or some accelerated variants \cite{Beck2009,Nesterov2007}, as presented in Algorithm~\ref{Algo:FBSAcc}. The gradient of the data-fidelity term $\mathcal{D}$ is given by
	\begin{equation}\label{eq:propo}
		\nab \mathcal{D}(\fd) = \sum_{p=1}^P \nab \mathcal{D}_p(\fd),
	\end{equation}
	with
	\begin{equation}
		\nab \mathcal{D}_p(\fd) = \Re\left( \Jd_{h_p}^H(\fd)  \tilde{\Gd}^H(\tilde{\Gd} \, \diag{\fd}\ud_p(\fd) - \yd_p^{\mathrm{sc}}) \right),
	\end{equation}
	where $\Jd_{h_p}(\fd)$ denotes the Jacobian matrix of  
	\begin{equation}\label{eq:funcH}
		h_p : \fd \mapsto \diag{\fd}\ud_p(\fd).
	\end{equation}	
Algorithm~\ref{Algo:FBSAcc} encompasses FISTA~\cite{Beck2009}  for a specific choice of the sequence $(\alpha^k)_{k\in \N}$. Its convergence is guaranteed in the convex case when $\gamma < 1/\mathrm{Lip}(\nab \mathcal{D})$, where $\mathrm{Lip}(\mathcal{\nab \mathcal{D}})$ is the Lipschitz constant of  $\nab \mathcal{D}$. In the nonconvex case, a local convergence of the classical FBS algorithm can be shown~\cite{Attouch2013}. Although, to the best of our knowledge, there exists no theoretical proof of convergence of accelerated versions for nonconvex function, Algorithm~\ref{Algo:FBSAcc} always converged in our experiments.

	\begin{algorithm}[t]
		\caption{Accelerated forward-backward splitting.}\label{Algo:FBSAcc}
	\begin{algorithmic}[1]
		\REQUIRE  $\fd^0 \in \R^N$, $(\alpha^k)_{k\in \N}$, $\gamma \in \left( 0, 1/\mathrm{Lip}(\nab \mathcal{D}) \right) $
		\STATE $\vd^1 = \fd^0$
		\STATE $k=1$
		\WHILE{(not converged)}
			\STATE $\ud_p^k \leftarrow \ud_p(\fd^k), \forall p \in [1\ldots P]$ (forward model \eqref{eq:ForwardOptPb})
			\STATE $\mathrm{\mathbf{d}}^k = \sum_{p=1}^P \Re\left( \Jd_{h_p}^H(\fd^k) \tilde{\Gd}^H(\tilde{\Gd} \, \diag{\fd^k}\ud_p^k - \yd_p^{\mathrm{sc}}) \right)  $
			\STATE $\fd^{k}= \mathrm{prox}_{\gamma \lambda \mathcal{R}}\left( \vd^k  - \gamma \mathrm{\mathbf{d}}^k\right) $
			\STATE $\vd^{k+1}=  \fd^k + \alpha^k (\fd^k - \fd^{k-1})$
			\STATE $k=k+1$
		\ENDWHILE
 	\end{algorithmic}
	\end{algorithm}

\subsubsection{Computation of $\Jd^H_{h_p}(\fd)$} 

The computation of $\Jd^H_{h_p}(\fd)$, required at line 5 of Algorithm~\ref{Algo:FBSAcc}, is challenging. The existence of a closed-form solution is made unlikely by the fact that the forward model in~\eqref{eq:LippScwingDiscret} itself requires one to invert an operator. We distinguish two distinct strategies.
	\begin{enumerate}
		\item SEAGLE: Build an error-backpropagation rule from the NAGD algorithm used to compute the forward model~\eqref{eq:ForwardOptPb}.
		\item Ours: Derive an explicit expression of $\Jd_{h_p}(\fd)$, as given in Section~\ref{sec:grad} (Proposition~\ref{propo}).
	\end{enumerate}
	
\subsubsection{Computation of $\mathrm{prox}_{\gamma \lambda \mathcal{R}}$} 
Numerous methods have been proposed to compute the proximity operator of $\mathcal{R}$, \cite{Beck2009b,Chambolle2011,Kamilov2017}. In SEAGLE, Liu \textit{et al.} use the algorithm proposed by Beck and Teboulle \cite{Beck2009b}.  Here, we compute it using the popular alternating-direction method of multipliers (ADMM)~\cite{Boyd2011}, which is well suited to the minimization of the sum of three convex functions. Moreover, it provides a high modularity for spatial regularization since one can easily change from one regularizer  (\eg TV) to another  (\eg Hessian Shatten-norm~\cite{Lefkimmiatis2013}). Details about the computation of  $\mathcal{R}$ are provided in Appendix~\ref{apndx:Prox}. 

\subsubsection{Speedup strategies} \label{sec:speedup}
 The cost of evaluating the forward model with \eqref{eq:ForwardOptPb} and the gradient $\nab \mathcal{D}$ is proportional to the number $P$ of illuminations $\ud_p^{in}$. However, these computations can easily be parallelized by performing the computation for each illumination (or each element of the sum in \eqref{eq:propo}) on a separate thread.  Moreover, in the spirit of the stochastic gradient-descent algorithm~\cite{Bottou2010}, we approximate $\nab \mathcal{D}$ as
	\begin{equation}
		\nab \mathcal{D}(\fd) \simeq \sum_{p \in \omega} \nab \mathcal{D}_p(\fd),
	\end{equation}
where $\omega$ is a subset of $[1\ldots P]$. We change $\omega$ at each iteration. Such a method is known to spare many computations when  $\nab \mathcal{D}_p$ does not admit a simple-form expression.  
	
\section{Efficient computation of the gradient $\nab \mathcal{D}$} \label{sec:grad}  
	 
	The error-backpropagation strategy used in SEAGLE to compute $\Jd^H_{h_p}(\fd)$ implies that one must store all the forward iterates. This consumes memory resources and compromises the deployment of the method for large 3D data. Instead, Proposition~\ref{propo} reveals that its computation requires one to invert the operator $(\Id -  \diag{\fd} \Gd^H)$. This operator has the same form (and size) that the operator we invert within the forward computation in \eqref{eq:ForwardOptPb}  and both can be computed in a similar way, using an iterative algorithm.  Moreover, it allows us to decouple the forward and gradient computation in Algorithm~\ref{Algo:FBSAcc}, which has the two following advantages:
\begin{itemize}
	\item choice of any iterative algorithm for computing \eqref{eq:ForwardOptPb} at line 4 of Algorithm~\ref{Algo:FBSAcc}, and computing $\Jd^H_{h_p}(\fd)$ at line 5 of Algorithm~\ref{Algo:FBSAcc} (see Section~\ref{sec:ForwardComputation});
	\item reduction of the memory consumption (no needs for storing forward iterates).
\end{itemize}

\begin{proposition} \label{propo} The Jacobian matrix of the function $h_p$ in \eqref{eq:funcH}  is given by 
\begin{equation}
	\Jd_{h_p}(\fd) = \left(  \Id +\diag{\fd} (\Id - \Gd \, \diag{\fd})^{-1}\Gd \right)  \diag{\ud_p(\fd)} .
\end{equation}
\end{proposition}
\begin{proof}
	
	We use the G\^ateaux derivative in the direction $\vd \in \R^N$ given by
	\begin{align}
		\dd h_p(\fd;\vd) & = \lim_{\varepsilon \rightarrow 0 } \frac{\diag{\fd + \varepsilon \vd}\ud_p(\fd + \varepsilon \vd) - \diag{\fd} \ud_p(\fd)}{\varepsilon} \notag\\
		&  = \diag{\ud_p(\fd)} \vd +  \lim_{\varepsilon \rightarrow 0 } \diag{\fd}\frac{\ud_p(\fd + \varepsilon \vd) - \ud_p(\fd)}{\varepsilon}  .\label{eq:proofPropo2}
	\end{align}
	
	Then, from \eqref{eq:LippScwingDiscret}, we get that
	\begin{align}\label{eq:proofPropo11}
		\ud_p^{\mathrm{in}} = & (\Id - \Gd \, \diag{\fd + \varepsilon \vd})\ud_p(\fd + \varepsilon \vd) \notag  \\
		  =  & (\Id - \Gd \, \diag{\fd})\ud_p(\fd + \varepsilon \vd) - \varepsilon\Gd \, \diag{ \vd} \ud_p(\fd + \varepsilon \vd) 
	\end{align}
	and
		\begin{equation}	\label{eq:proofPropo12}
		 (\Id - \Gd \,\diag{\fd})\ud_p(\fd) =  \ud_p^{\mathrm{in}}.
\end{equation}
	Combining \eqref{eq:proofPropo11} and \eqref{eq:proofPropo12}, we obtain that
	\begin{equation}\label{eq:proofPropo}
		 (\Id - \Gd \, \diag{\fd})(\ud_p(\fd + \varepsilon \vd)   - \ud_p(\fd) ) = \varepsilon\Gd \, \diag{ \vd} \ud_p(\fd + \varepsilon \vd) .
	\end{equation}
	Finally, we get that
	\begin{equation}
		\dd h_p(\fd;\vd)  =  \left(  \Id +\diag{\fd} (\Id - \Gd \, \diag{\fd})^{-1}\Gd \right)  \diag{\ud_p(\fd)} \vd
	\end{equation}
	and, thus, that
	\begin{equation}
	  \Jd_{h_p}(\fd)  = \left(  \Id +\diag{\fd} (\Id - \Gd \, \diag{\fd})^{-1}\Gd \right)  \diag{\ud_p(\fd)},
\end{equation}
	which completes the proof.
\end{proof}

\section{Algorithm analysis}\label{sec:AlgoAnalysis}

\subsection{Memory requirement} 

	In this section, we elaborate on the memory consumption of the proposed method in comparison with SEAGLE.  First, let us state that gradient based methods, such as NAGD or CG, have similar memory requirements. It corresponds roughly to three times the size of the optimization variable which is the part that is common to both algorithms. The additional memory requirement that is specific to SEAGLE relies only on the storage of the NAGD iterates during the forward computation. Suppose that $K_{\mathrm{NAGD}} \in \N$ iterations are necessary to compute the forward model with \eqref{eq:ForwardOptPb} and that the region $\Omega$ is sampled over $N \in \N$ pixels (voxels, in 3D). Since the total field $\ud_p(\fd)$ computed by NAGD is complex-valued, each pixel is represented with 16 bytes (double precision for accurate computations). Hence, the difference of memory consumption between SEAGLE and our method is	
	\begin{equation}
		\Delta_{\mathrm{Mem}}= N\times K_{\mathrm{NAGD}}  \times 16 \;[\mathrm{bytes}],
	\end{equation}
	which corresponds to the storage of the $K_{\mathrm{NAGD}}$ intermediate iterates of NAGD.
	Here, we assumed that $\nab \mathcal{D}$ was computed by sequentially adding the partial gradients $\nab \mathcal{D}_p$ associated to the $P$ incident fields. Hence, once the partial gradient associated to one incident angle is computed by successively applying  the forward model (NAGD) and  the error-backpropagation procedure, the memory used to store the intermediate iterates can be recycled to compute the partial gradient associated to the next incident angle. However, when the parallelization strategy detailled in Section \ref{sec:speedup} is used, the memory requirement is mutiplied by the number $N_{\mathrm{Threads}}\in \N$ of threads, so that
	\begin{equation}\label{eq:memory}
		\Delta_{\mathrm{Mem}}= N\times K_{\mathrm{NAGD}} \times N_{\mathrm{Threads}} \times 16 \; [\mathrm{bytes}].
	\end{equation}
	Indeed, since the threads of a single computer share  memory, computing $N_{\mathrm{Threads}}$ partial gradients in parallel requires $N_{\mathrm{Threads}}$ times more memory.
	
	For illustration, we give in Fig.~\ref{fig:Memory} the evolution of $\Delta_{\mathrm{Mem}}$ as a function of $N$ for different values of $K_{\mathrm{NAGD}}$ and $N_{\mathrm{Threads}}$. One can see with the vertical dashed lines that, for 3D volumes,  the memory used by SEAGLE quickly reaches several tens of Megabytes, even for small volumes (\eg $128\times 128 \times 128$), to hundreds of Gigabytes for the larger volumes that are typical of microscopy (\eg $512\times 512 \times 256$). This shows the limitation of SEAGLE for 3D reconstruction in the presence of a shortage of memory resources and reinforces the interest of the proposed alternative.
	
	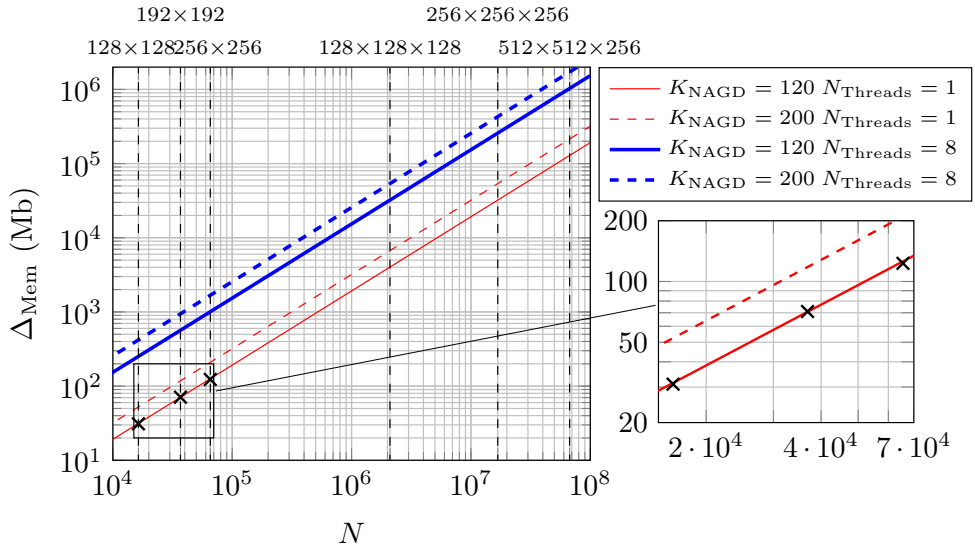
\begin{figure}[t]
		\centering 
		\resizebox {\textwidth} {!} {
		\begin{tikzpicture}[every pin/.style={rectangle,rounded corners=3pt,font=\tiny},]
		\begin{axis}[
    					  legend style={at={(1.8,1)},legend cell align=left},
						  grid=both,                        
						  ymode=log,                        
						  xmode=log,                        
						  xlabel={$N$},
						  ylabel={$\Delta_{\mathrm{Mem}}$ (Mb)},
						  xlabel style={yshift=0.2cm},
						  xmax=100000000,ymax=2000000,xmin=10000,ymin=10,
						  extra x ticks={16384,36864,65536,2097152,16777216,67108864},
        			     extra x tick labels={{\scriptsize 128$\times$128 $\;$},{\scriptsize 192$\times$192 \\ $ $},{\scriptsize $\;$ 256$\times$256},{\scriptsize 128$\times$128$\times$128},{\scriptsize 256$\times$256$\times$256 \\ $ $},{\scriptsize 512$\times$512$\times$256}},
        			     extra x tick style={%
            				 ,ticklabel pos=upper
            				 ,align=center
          				 },
    					  width=0.55\textwidth]		
    		\addplot[red,thin,domain=1:100000000,samples=400] {x*120*1*16/10^6};		
    		\addplot[red,dashed, thin,domain=1:100000000,samples=400] {x*200*1*16/10^6};
    		\addplot[blue,very thick,domain=1:100000000,samples=400] {x*120*8*16/10^6};
    		\addplot[blue,dashed, very thick,domain=1:100000000,samples=400] {x*200*8*16/10^6};	
    		\addplot[black,mark=x,mark size=3,thick] coordinates {(16384,31)};               
    		\addplot[black,mark=x,mark size=3,thick] coordinates {(36864,71)};               
    		\addplot[black,mark=x,mark size=3,thick] coordinates {(65536,123)};           
    		\addplot[black,dashed] coordinates {(16384,1) (16384,2000000)};               
    		\addplot[black,dashed] coordinates {(36864,1) (36864,2000000)};               
    		\addplot[black,dashed] coordinates {(65536,1) (65536,2000000)};               
    		\addplot[black,dashed] coordinates {(2097152,1) (2097152,2000000)};       
    		\addplot[black,dashed] coordinates {(16777216,1) (16777216,2000000)};   
    		\addplot[black,dashed] coordinates {(67108864,1) (67108864,2000000)};   
    		\draw  (axis cs:15000,20) rectangle (axis cs:70000,200);
			\legend{{\scriptsize $K_{\mathrm{NAGD}}=120 \; N_{\mathrm{Threads}}=1$},{\scriptsize $K_{\mathrm{NAGD}}=200 \; N_{\mathrm{Threads}}=1$},{\scriptsize $K_{\mathrm{NAGD}}=120 \; N_{\mathrm{Threads}}=8$},{\scriptsize $K_{\mathrm{NAGD}}=200 \; N_{\mathrm{Threads}}=8$}};
		\end{axis}		
			\node[pin={[pin edge=black,pin distance=5.3cm]10:{}}] at (1.1,0.8) {};
	   \node at (8,1.5) {\begin{tikzpicture}[baseline,trim axis left,trim axis right]
	    \begin{axis}[
	 					  width=0.35\textwidth,
	      				  grid=both,                        
						  ymode=log,                        
						  xmode=log,                        
	     				  xmax=70000,ymax=200,xmin=15000,ymin=20,
	     				  xtick={20000,30000,40000,50000,60000,70000},
	     				  ytick={20,30,40,50,60,70,80,90,100,200},
	     				  xticklabels={$2\cdot10^4$,,$4\cdot10^4$,,,$7\cdot10^4$},
	     				  yticklabels={20,,,50,,,,,100,200},
	    				 ]
	    \addplot[red,thick,domain=1:70000,samples=400] {x*120*1*16/10^6};		
	    		\addplot[red,dashed, thick,domain=1:70000,samples=400] {x*200*1*16/10^6};
	    		\addplot[black,mark=x,mark size=3,thick] coordinates {(16384,31)};               
	    		\addplot[black,mark=x,mark size=3,thick] coordinates {(36864,71)};               
	    		\addplot[black,mark=x,mark size=3,thick] coordinates {(65536,123)};      
	    \end{axis}
	    \end{tikzpicture}};
	\end{tikzpicture}
	}
	\caption{\label{fig:Memory} Predicted evolution of $\Delta_{\mathrm{Mem}}$ as function of the number $N$ of points for two values of  $K_{\mathrm{NAGD}}$ and $N_{\mathrm{Threads}}$. The vertical dashed  lines give examples of 2D and 3D volumes for a range of values of $N$. Finally, the three crosses correspond to  values of $\Delta_{\mathrm{Mem}}$ measured experimentally.}
\end{figure}

\subsection{Conjugate gradient \textit{vs.} Nesterov accelerated gradient descent for \eqref{eq:ForwardOptPb}} \label{sec:ForwardComputation}  

	Due to Proposition~\ref{propo}, we can compute both \eqref{eq:ForwardOptPb} and $\Jd^H_{h_p}(\fd)$ using any state-of-the-art quadratic optimization algorithm. This contrasts with SEAGLE, where one must derive the error-backpropagation rule from the forward algorithm, which may limit its choice. We now provide numerical evidence that GC is more efficient than NAGD for solving \eqref{eq:ForwardOptPb}. To this end, we consider a circular object (bead) of radius $r_{\mathrm{bead}}$ and refractive index $n_{\mathrm{bead}}$ immersed into water ($n_{\mathrm{b}}=1.333$), as presented in Fig.~\ref{fig:ExForward} (top-left).  In such a situation, an analytic expression of the total field is provided by the  Mie theory \cite{Devaney2012,Stratton2007}. Hence, at each iteration $k$, we compute the relative error $\varepsilon_k$ of the current estimate $\ud^k$  to the Mie solution~$\ud_{\text{\tiny Mie}}$ as
\begin{equation}\label{eq:relativeErrForward}
	\varepsilon_{k} = \frac{\|\ud^k - \ud_{\text{\tiny Mie}}\|^2}{\|\ud_{\text{\tiny Mie}}\|^2}.
\end{equation}
	In our experiment, the bead is impinged by a plane wave of wavelength $\lambda=406$~\nano\meter. The region of interest is square with a side length of $16 \lambda$ (see top-left panel of Fig.~\ref{fig:ExForward}). It is sampled using $1,\!024$ points along each side. We used a fine grid in order to limit the impact of numerical errors related to discretization. The wave source corresponds to the bottom border of this region. Then, as in \cite{Liu2016,Liu2017}, we refer to the refractive index $n_{\mathrm{bead}}$ by its contrast with respect to the background medium, defined as $\max(|\fd|)/(k_0^2 n_{\mathrm{b}}^2)$. We show in Fig.~\ref{fig:ForwardComparisons}  the evolution of $k_{\varepsilon_0}$, which is the number of iterations needed to let the relative error \eqref{eq:relativeErrForward} fall below $\varepsilon_0=10^{-2}$. One can observe that  CG  is much more efficient than NAGD, in particular for large contrasts. This is not negligible since an evaluation of the forward model is required at each iteration of Algorithm~\ref{Algo:FBSAcc} (line 4). Our comparison in terms of a number of iterations is fair because  the computational cost of one iteration is the same for both algorithms. Note that the descent step of NAGD was adapted during the iterations following the same rule as in~\cite{Liu2016,Liu2017}.

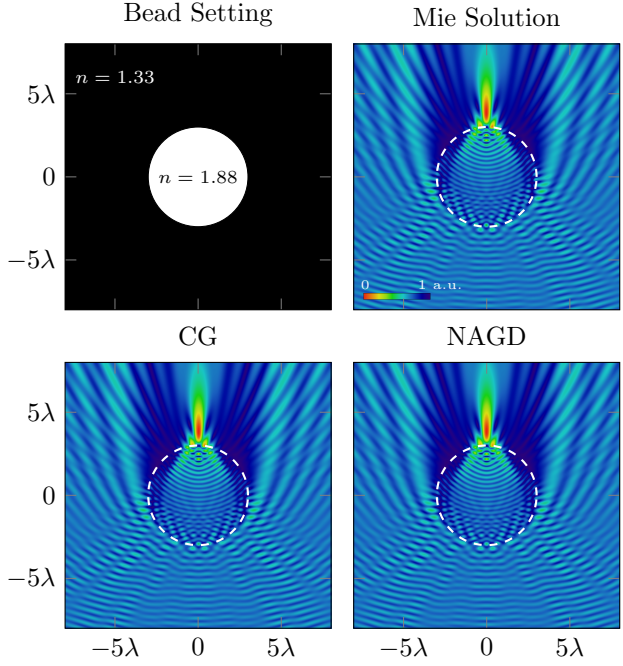
\begin{figure}[t]
		\centering
				\begin{tikzpicture}
		\begin{groupplot}[group style={group size= 2 by 2,                      
    					  horizontal sep=0.3cm, vertical sep=0.7cm},          
						  xmin=-8,xmax=8,
					   	  ymin=-8,ymax=8,
						  title style={yshift=-0.10cm},
						  axis equal image,
						  axis on top,
						  grid style={black},
    					  width=0.45\textwidth]
    	\nextgroupplot[title={Bead Setting},xticklabels={,,},ytick={-5,0,5},yticklabels={$-5\lambda$,$0$,$5\lambda$}] 	
    	\draw [fill=black] (axis cs:-8,-8) rectangle (axis cs:8,8);
    	\draw [fill=white] (axis cs:0,0) circle (3);
    	\node[black] at (axis cs:0,0) {{\scriptsize $n=1.88$}};
    	\node[white] at (axis cs:-5,6) {{\scriptsize $n=1.33$}};
    	\nextgroupplot[enlargelimits=false,title={Mie Solution},yticklabels={,,},xticklabels={,,}] 	
    	   \addplot[] graphics[xmin=-8,ymin=-8,xmax=8,ymax=8] {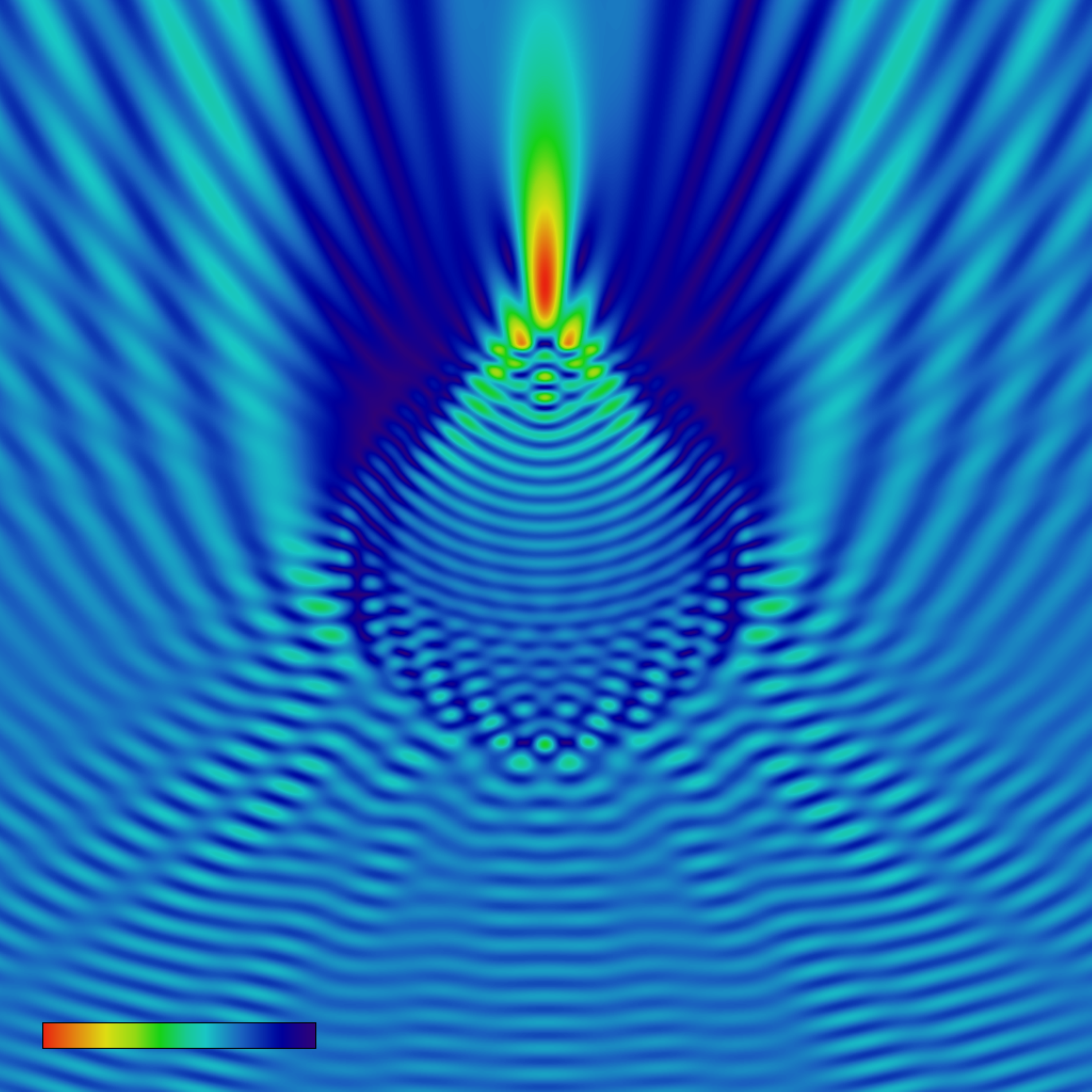};
    		\draw [dashed,white,thick] (axis cs:0,0) circle (3);   		
    		\node[white] at (axis cs:-7.3,-6.5) {{\tiny0}};
    		\node[white] at (axis cs:-2.8,-6.5) {{\tiny 1 a.u.}};
    	\nextgroupplot[enlargelimits=false,title={CG},ytick={-5,0,5},yticklabels={$-5\lambda$,$0$,$5\lambda$},xtick={-5,0,5},xticklabels={$-5\lambda$,$0$,$5\lambda$}] 			  
    		\addplot[] graphics[xmin=-8,ymin=-8,xmax=8,ymax=8] {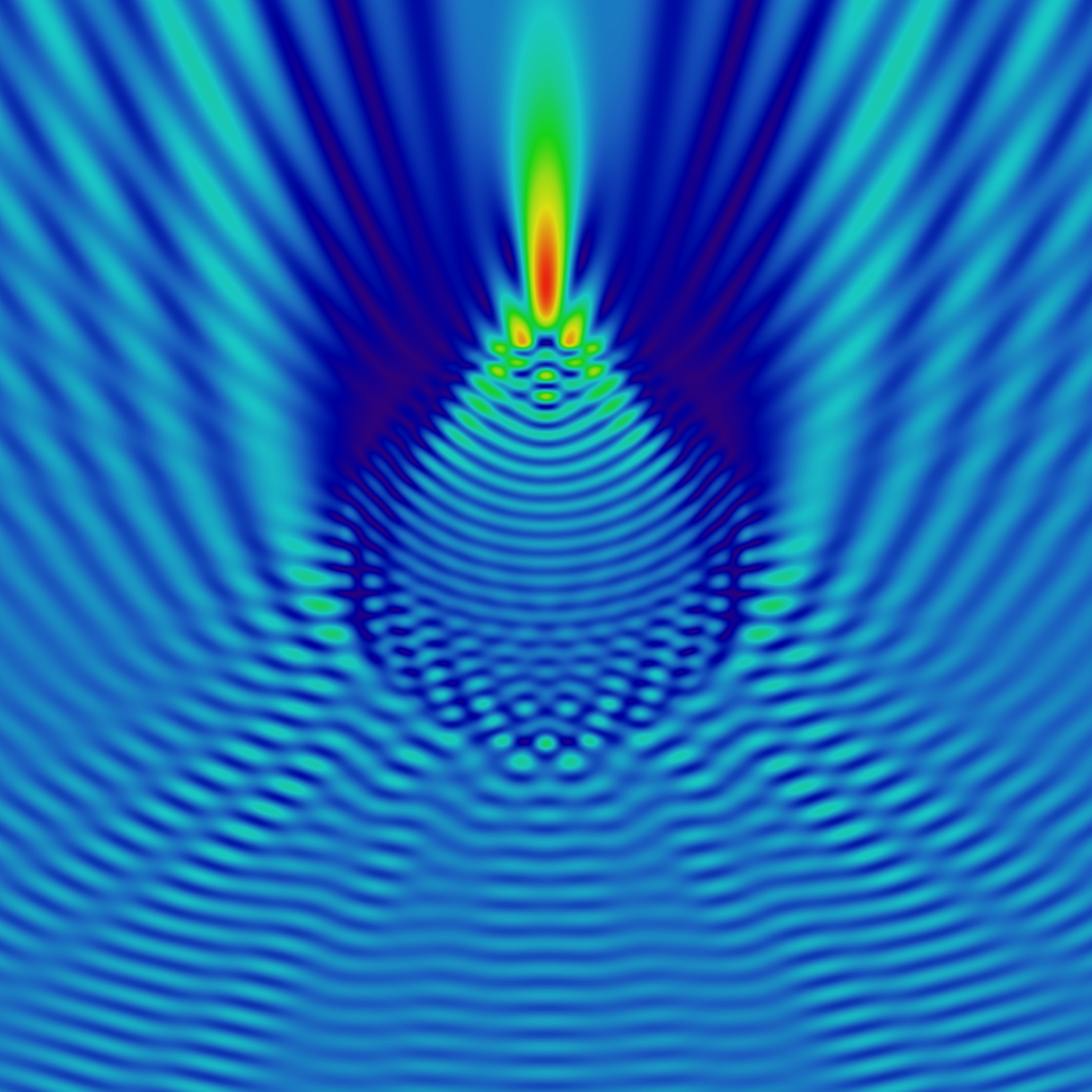};
    		\draw [dashed,white,thick] (axis cs:0,0) circle (3);
    	\nextgroupplot[enlargelimits=false,title={NAGD},yticklabels={,,},xtick={-5,0,5},xticklabels={$-5\lambda$,$0$,$5\lambda$}] 			  
    		\addplot[] graphics[xmin=-8,ymin=-8,xmax=8,ymax=8] {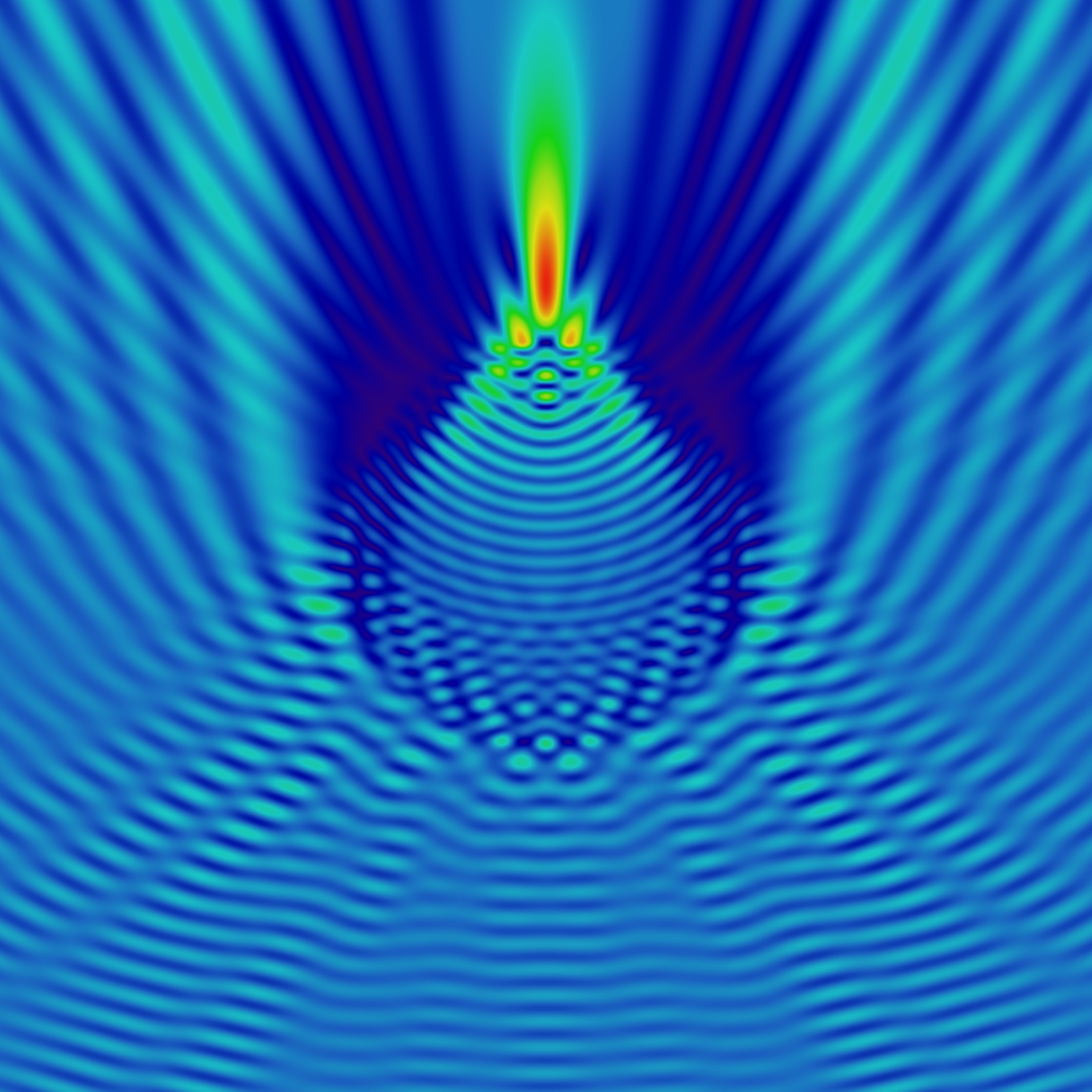};
    		\draw [dashed,white,thick] (axis cs:0,0) circle (3);
		\end{groupplot}
	\end{tikzpicture}
		\caption{\label{fig:ExForward} Forward-model solution for a bead with radius $3\lambda$ and a contrast of $1$ using CG (bottom-left) and NAGD (bottom-right), as well as the Mie solution (top-right). The setting used for this experiment is presented in the top-left panel. The colormap is the same for each figure.}
	\end{figure}

\begin{figure}[t]
		\centering 
		\begin{tikzpicture}
		\begin{groupplot}[group style={group size= 2 by 1,horizontal sep=0.3cm},
    					  legend pos= south east,           
    					  legend style={legend cell align=left},
						  grid=both,                        
						  ylabel absolute, 
						  ymode=log,                        
    					  width=0.45\textwidth]		
    		\nextgroupplot[ylabel={$k_{\varepsilon_0}$},xlabel={$r_{\mathrm{bead}}$ ($\lambda$ unit)},ymin=10,ymax=3000,xmin=2,xmax=4,title={Contrast $0.3$},xtick={2,2.5,3,3.5,4},xticklabels={$2\lambda$, ,$3\lambda$, ,$4\lambda$},ylabel absolute, ylabel style={yshift=-0.3cm}]	
    		\addplot[red,very thick,mark=diamond,mark size=4] table{N_1024_Contrast_0.3_varRadius_res_CG_final.dat};
			\addplot[blue,very thick,mark=o,mark size=3] table{N_1024_Contrast_0.3_varRadius_res_Nest_final.dat};
			\nextgroupplot[xlabel={Contrast},ymin=10,ymax=3000,xmin=0.1,xmax=1,title={$r_{\mathrm{bead}}=3\lambda$},yticklabels={ , , }]	
    		\addplot[red,very thick,mark=diamond,mark size=4] table{Radius_3_N_1024_varContrast_res_CG_final.dat};
			\addplot[blue,very thick,mark=o,mark size=3] table{Radius_3_N_1024_varContrast_res_Nest_final.dat};
			\legend{CG,NAGD};
		\end{groupplot}
	\end{tikzpicture}
	\caption{\label{fig:ForwardComparisons} Evolution of the number of iterations $k_{\varepsilon_0}$ needed to let the relative error \eqref{eq:relativeErrForward} fall below $\varepsilon_0=10^{-2}$ as function of bead radius (left) and bead contrast (right).}
\end{figure}
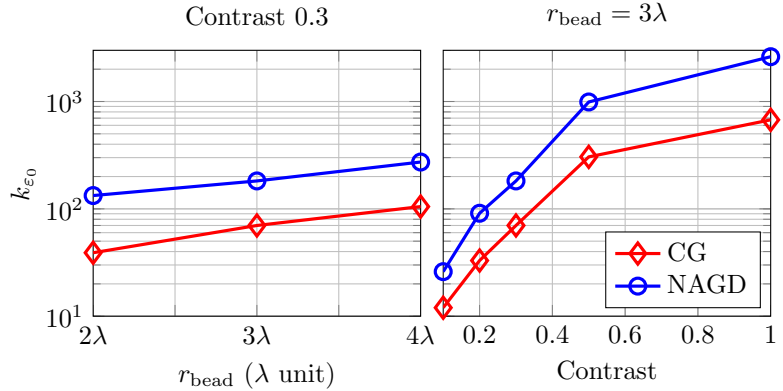

Finally, the solution obtained with the two algorithms for $r_{\mathrm{bead}}=3\lambda$ and a contrast of $1$ are shown in Fig.~\ref{fig:ExForward}. The analytic Mie solution is also provided for comparison. From these figures, one can appreciate the high accuracy obtained by solving \eqref{eq:ForwardOptPb}, as first demonstrated in \cite{Liu2016,Liu2017}.

\section{Numerical experiments} \label{sec:ExpNum}

	This section is devoted to numerical experiments that illustrate the two main advantages of the proposed method over SEAGLE, which consist in a reduced computational cost and a reduced memory consumption. The algorithms have been implemented using an inverse-problem library developed in our group \cite{Unser2017} (\textit{GlobalBioIm}: http://bigwww.epfl.ch/algorithms/globalbioim/). Hence, they share the implementation of the overall FISTA algorithm as well as inner procedures such as the computation of the proximity operator of $\mathcal{R}$ (see Appendix~\ref{apndx:Prox}). The only difference between the two methods resides in the computation of the forward model in \eqref{eq:ForwardOptPb} and  $\Jd^H_{h_p}(\fd)$. For SEAGLE, this is performed using the NAGD algorithm and an error-backpropagation strategy. For our method, \eqref{eq:ForwardOptPb} and $\Jd^H_{h_p}(\fd)$ are computed using the CG algorithm, in accordance with Proposition \ref{propo}. Note that no parallelization is used. Reconstructions are performed with  MATLAB 9.1 (The MathWorks Inc., Natick, MA, 2000) on a PowerEdge T430 Dell computer (Intel Xeon~E5-2620~v3).
	
\subsection{Simulated data}\label{sec:ExpNumSynth}
	
\subsubsection{Simulation settings}

	The Shepp-Logan phantom of Fig.~\ref{fig:shepp} has the contrast  $\max(|\fd|)/(k_0^2 n_{\mathrm{b}}^2)=0.2$. It is immersed into water ($n_{\mathrm{b}}=1.333$). The wavelength of  the incident plane waves  is  $\lambda=406$~\nano\meter. We consider thirty-one incident angles, from $-60$\degree~to $+60$\degree. The sources are placed at the bottom side of the sample, at a distance of 16.5$\lambda$ from its center. Moreover, we consider two detectors placed on both top and bottom sides of the object, also at a distance of 16.5$\lambda$ from  its center.  Hence, the overall region is a square of length $33\lambda$ per side. Data are simulated using a fine discretization  of this region, with a ($1024 \times 1024$) grid that leads to square pixels of surface $(3.223\cdot 10^{-2}\lambda)^2$. We used a large number of CG iterations to get the accurate simulation mentioned in Section~\ref{sec:ForwardComputation}.  Then, the measurements were extracted from the first and last rows of each total field associated to the incident fields. This lead to a total of ($31 \times 2 \times 1024$) measurements. Finally, we defined three ODT problems by downsizing (using averaging) the ($31 \times 2 \times 1024$) measurements to grids with size of ($31 \times 2 \times 512$), ($31 \times 2 \times 384$), and ($31 \times 2 \times 256$). 
	
	This setting corresponds to an ill-posed and highly scattering situation. Moreover, the detector length is only two times larger than the object, which results in a loss of information for large incident angles.  This makes the resulting inverse problem challenging.
	
	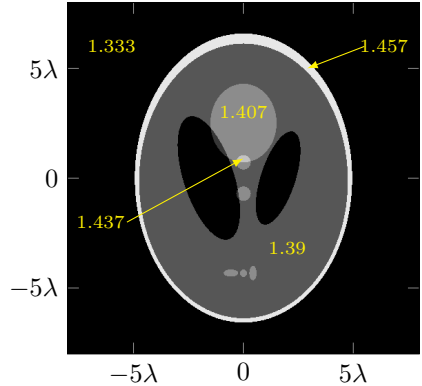
\begin{figure}[t]
		\centering
		\begin{tikzpicture}
		\begin{axis}[	  xmin=-8,xmax=8,
					   	  ymin=-8,ymax=8,
						  axis equal image,
						  axis on top,
						  grid style={black},
						  enlargelimits=false,
						  ytick={-5,0,5},yticklabels={$-5\lambda$,$0$,$5\lambda$},
						  xtick={-5,0,5},xticklabels={$-5\lambda$,$0$,$5\lambda$},
    					  width=0.55\textwidth]
    	  \addplot[] graphics[xmin=-8,ymin=-8,xmax=8,ymax=8] {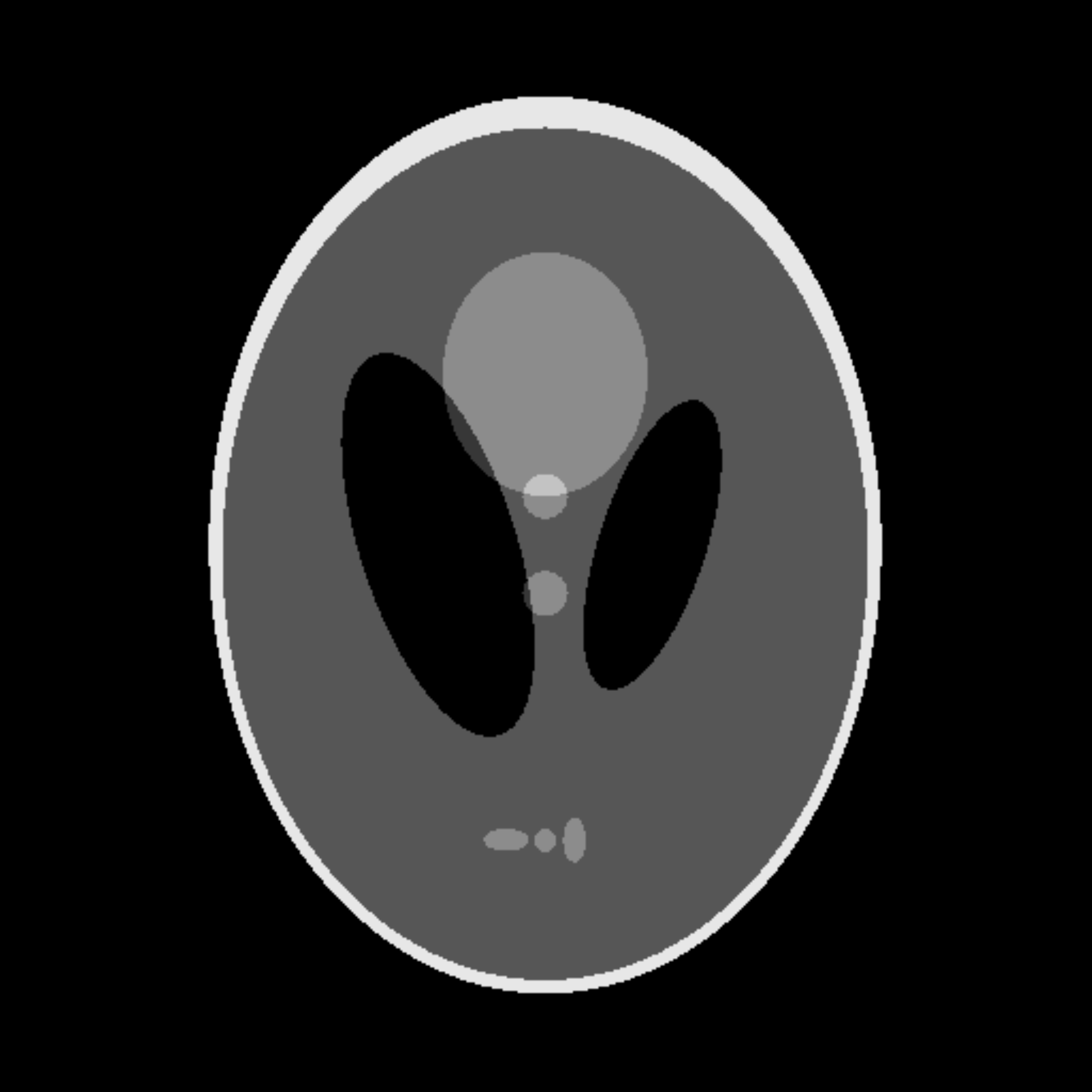};
    	  \node[yellow] at (axis cs:-6,6) {{\scriptsize 1.333}};
    	  \node[yellow] at (axis cs:2,-3.2) {{\scriptsize 1.39}};
    	  \node[yellow] at (axis cs:6.4,6) {{\scriptsize 1.457}};
    	  \node[yellow] at (axis cs:0,3) {{\scriptsize 1.407}};
    	  \node[yellow] at (axis cs:-6.5,-2) {{\scriptsize 1.437}};
    	  \draw[yellow,-latex] (axis cs:5.5,6) -- (axis cs:2.9,5);
    	  \draw[yellow,-latex] (axis cs:-5.3,-2) -- (axis cs:0,0.9);
		\end{axis}
		\end{tikzpicture}
		\caption{\label{fig:shepp} Sheep-Logan phantom and refractive indices of the gray levels. The contrast is 20\%.}
	\end{figure}

\subsubsection{Algorithm parameters}

	For each simulated OTD problem, we considered a square region of interest $\Omega$ with sides half the sources--detector distance. That corresponds to images of size ($256 \times 256$) with pixels of area $(6.445\cdot 10^{-2} \lambda)^2$, ($192 \times 192$) with pixels of area $(8.839\cdot 10^{-2} \lambda)^2$, and ($128 \times 128$) with pixels of area $(1.289\cdot 10^{-1} \lambda)^2$. The support of the phantom is fully contained in $\Omega$. 
	
	Then, to compute the gradient (stochastic-gradient strategy), we  selected eight angles over the thirty-one that were available and changed this selection at each iteration (see Section \ref{sec:speedup}).

The NAGD or CG forward algorithms are stopped either after hundred-twenty iterations or when the relative error between two iterates is below $10^{-4} $. Finally,   two-hundred iterations of FISTA are performed with a descent step   fixed empirically  to $\gamma = 5\cdot 10^{-3}$. We used the regularization parameter $\mu= 3.3 \cdot 10^{-2}$.

\subsubsection{Metrics} 

 We compared the two methods in terms of running time and memory consumption, as measured by the peak memory (maximum allocated memory) reached by each algorithm during execution. The outcome is reported in Table~\ref{table}. Once again, due to the use of our inverse-problem library \cite{Unser2017}, the comparison of the two methods is fair because  their implementations differ only by the forward algorithm and by the computation of $\Jd^H_{h_p}(\fd)$. Moreover, CG and NAGD  are implemented in the same fashion since they inherit the same optimization class of our inverse-problem library. Finally, we also provide the SNR of the reconstructed refractive index and observe that the computational gain comes at no cost in quality.

\subsubsection{Discussion}

	Our proposed alternative to SEAGLE allows us to reduce both time  and memory. Moreover, we have measured the peak memory difference $\Delta_{\mathrm{Mem}}$ between the two methods and superimposed it on the predictions of Fig.~\ref{fig:Memory}  where the adequacy between the theoretical curves and the measured points is remarkable. Hence, although our experiments  are restricted to 2D data, where the gap between the two algorithms is moderate, the evolution of $\Delta_{\mathrm{Mem}}$ for 3D data can be extrapolated from Fig.~\ref{fig:Memory}. This shows the relevance of our method when size increases.
	
	The SNR values given in Table \ref{table} as well as the reconstructions presented in Fig.~\ref{fig:Results} suggest that the two methods perform similarly in terms of quality. This is not surprising since the overall algorithm is the same, the differences residing merely in the computation of the forward model in \eqref{eq:ForwardOptPb} and the Jacobian  $\Jd_{h_p}(\fd)$. Moreover, one can observe that the quality of  reconstruction  decreases when the discretization grid becomes coarser. Indeed, the model is insufficiently accurate when the discretization is too poor. For instance, in the case of the ($128\times128$) grid, one wavelength unit is discretized using eight pixels, which is clearly detrimental to the accuracy of the forward model.
	
	Reconstructions for the ($256 \times 256$) problem are presented in Fig. \ref{fig:Results} for completeness. Besides the difficulty of the considered scenario, the two methods are able to retrieve most details of the object in comparison with the Rytov approximation. Artifacts are mainly due to the missing-cone problem and to the limited length of the detector. This corroborates the findings of \cite{Liu2017}.

	\begin{table}[t]
	\caption{\label{table} Proposed method \textit{vs.} SEAGLE \cite{Liu2016,Liu2017} in terms of running time and memory consumption. The reconstructed refractive-index maps are presented in Fig.~\ref{fig:Results}.}
		\centering		
		\begin{tabular}{lcccccc} 
			\toprule
			\toprule
			ROI $\Omega$ size  & \multicolumn{2}{c}{($128 \times 128$)} & \multicolumn{2}{c}{($192 \times 192$)} & \multicolumn{2}{c}{($256 \times 256$)} \\
			\cline{2-3} \cline{4-5} \cline{6-7}
			Method & Ours & \cite{Liu2017}  & Ours & \cite{Liu2017}    & Ours & \cite{Liu2017}   \\ \midrule
			 Time (\minute) & \textbf{9} & 35  & \textbf{12} & 72 & \textbf{19} & 110  \\
			 Memory (\mega b) & \textbf{138} & 169 & \textbf{224} & 295 & \textbf{337} & 460 \\
			 SNR (dB) & 43.96 & 43.76 & 45.44 & 45.48 & 46.96 & 46.99\\
			\bottomrule
			\bottomrule
		\end{tabular}		
	\end{table}

		\begin{figure}[t]
		\centering
				\begin{tikzpicture}
		\begin{groupplot}[group style={group size= 3 by 1,                      
    					  horizontal sep=0.5cm, vertical sep=0.5cm},          
						  xmin=-8,xmax=8,
					   	  ymin=-8,ymax=8,
						  axis equal image,
						  axis on top,
						  grid style={black},
    					  width=0.45\textwidth]
    	\nextgroupplot[enlargelimits=false,ylabel={},xtick={-5,0,5},xticklabels={$-5\lambda$,$0$,$5\lambda$},ytick={-5,0,5},yticklabels={$-5\lambda$,$0$,$5\lambda$},ylabel absolute,title={Rytov Approximation}, ylabel style={yshift=-0.25cm}] 			
    	\addplot[] graphics[xmin=-8,ymin=-8,xmax=8,ymax=8] {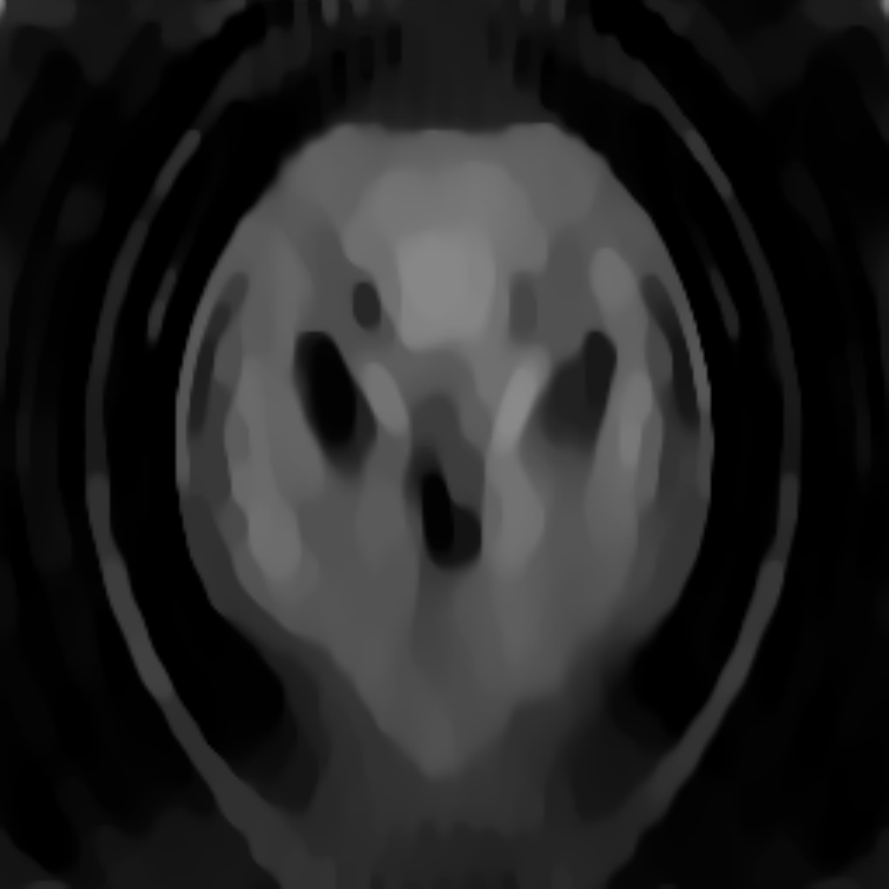};   	 	
    	\nextgroupplot[enlargelimits=false,ylabel={},xtick={-5,0,5},xticklabels={$-5\lambda$,$0$,$5\lambda$},ytick={-5,0,5},yticklabels={,,},ylabel absolute,title={Ours}, ylabel style={yshift=-0.25cm}] 			
    	\addplot[] graphics[xmin=-8,ymin=-8,xmax=8,ymax=8] {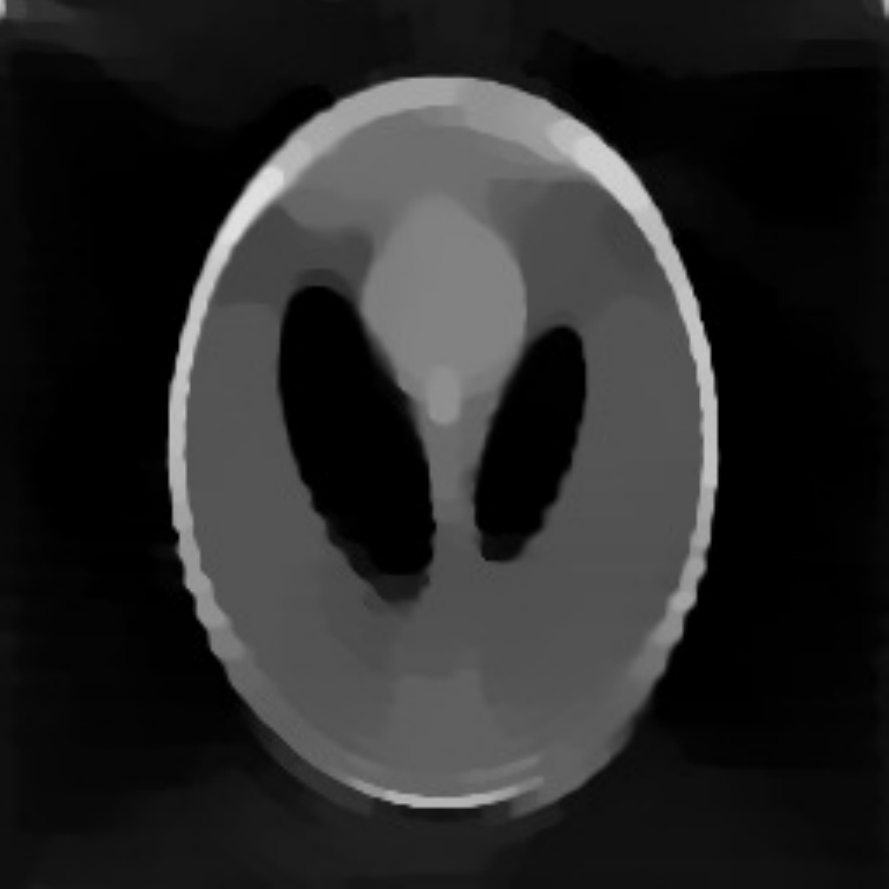};   	 	
    	\nextgroupplot[enlargelimits=false,xtick={-5,0,5},yticklabels={,,},xticklabels={$-5\lambda$,$0$,$5\lambda$},title={SEAGLE}] 	
    	\addplot[] graphics[xmin=-8,ymin=-8,xmax=8,ymax=8] {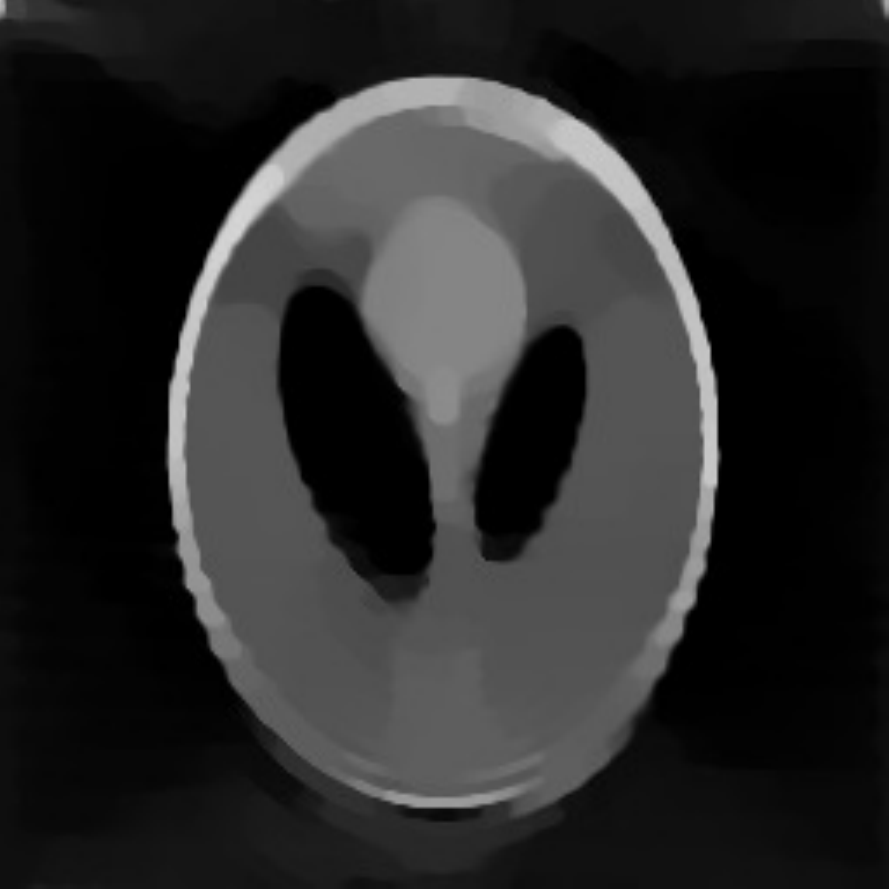};	   	
		\end{groupplot}
	\end{tikzpicture}
		\caption{\label{fig:Results} Reconstructions obtained by the proposed method and by SEAGLE for the $(256 \times 256)$ ODT problem with $\mu = 3.3\cdot 10^{-2}$. The colormap is identical to that of Fig.~\ref{fig:shepp}. For comparison, we provide the TV-regularized Rytov reconstruction with $\mu =3\cdot 10^{-3}$.}
	\end{figure}

\subsection{Real data}
	We evaluated our method using the \textit{FoamDielExt} target (TM polarisation) of the Institut Fresnel's public database~\cite{Geffrin2005}.  The data were collected for the two-dimensional inhomogeneous sample depicted  in the left panel of Fig.~\ref{fig:ResultsReal}. The permitivity of the ground truth was measured experimentally and is subject to uncertainties~\cite{Geffrin2005}.  The object is fully contained in a square region of length $15$ cm per side, which we discretize using a $256\times256$ grid. Sensors were distributed circularly around the object, at a distance of $1.67$ m from its center, and with a step of $1$\degree. Eight sources, uniformly distributed around the object, were sequentially activated. For each activated source, the sensors closer than $60$\degree~from the source were excluded. Thus, $241$ detectors among the $360$ available were used for each source. Frequencies from $2$ to $10$ GHz with a step of $1$ GHz are available in the database but we used only the $3$ GHz measurements (\ie $\lambda=10$ cm).
	
	The NAGD or CG forward algorithms are stopped either after two-hundred iterations or when the relative error between two iterates is below $10^{-6} $. Hundred iterations of FISTA are performed with a descent step  $\gamma = 5\cdot 10^{-3}$. We used the regularization parameter $\mu= 1.6 \cdot 10^{-2}$.
	
	In Fig.~\ref{fig:ResultsReal}, we see that both methods provide good reconstructions that are essentially indistinguishable (see also SNR values provided in the caption of the figure). This corroborates  the simulated numerical experiments of Section \ref{sec:ExpNumSynth}. The main point here is that, for this setting, the proposed method was 15 times faster than SEAGLE. 
	
	\begin{figure}[t]
		\centering
				\begin{tikzpicture}
		\begin{groupplot}[group style={group size= 3 by 1,                      
    					  horizontal sep=0.5cm, vertical sep=0.5cm},          
						  xmin=-0.75,xmax=0.75,
					   	  ymin=-0.75,ymax=0.75,
						  axis equal image,
						  axis on top,
						  grid style={black},
    					  width=0.45\textwidth]
    	\nextgroupplot[enlargelimits=false,ylabel={},xtick={-0.5,0,0.5},xticklabels={$-\frac{\lambda}{2}$,$0$,$\frac{\lambda}{2}$},ytick={-0.5,0,0.5},yticklabels={$-\frac{\lambda}{2}$,$0$,$\frac{\lambda}{2}$},ylabel absolute,title={Ground truth}, ylabel style={yshift=-0.25cm}] 			
    	\addplot[] graphics[xmin=-0.75,ymin=-0.75,xmax=0.75,ymax=0.75] {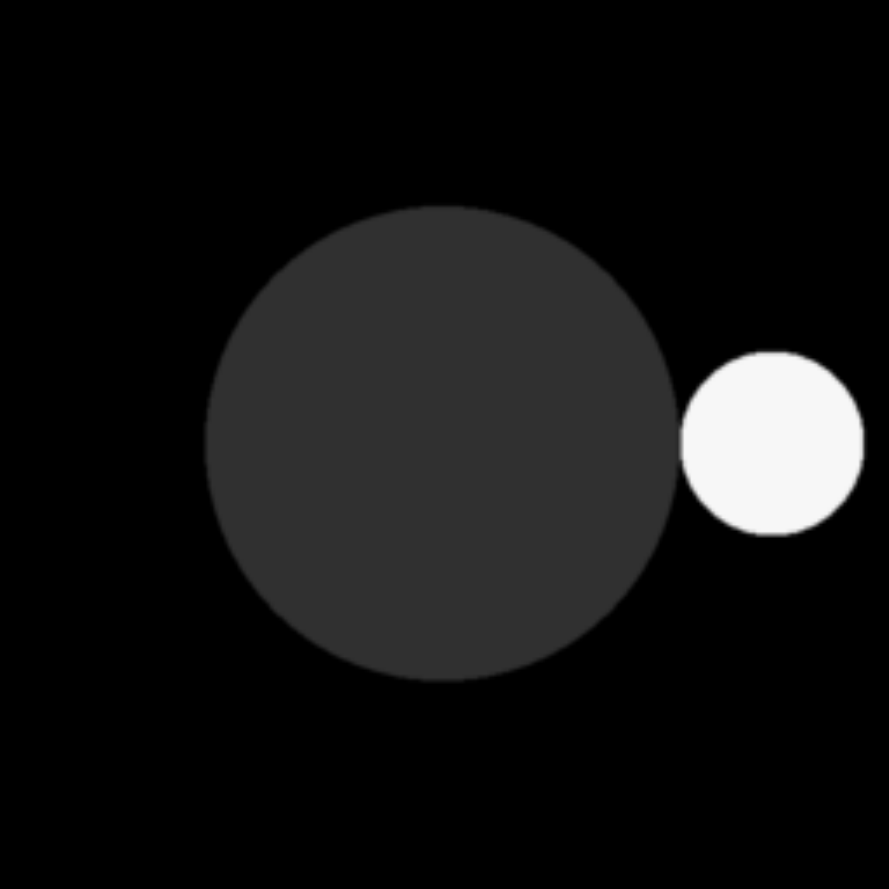};   	 
    	 \node[yellow] at (axis cs:0,0) {{\scriptsize $1.45 \pm 0.15$}};	
    	 \node[yellow] at (axis cs:-0.5,0.5) {{\scriptsize $1$}};	
    	 \node[yellow] at (axis cs:0.55,-0.45) {{\scriptsize $3 \pm 0.3$}};
    	  \draw[yellow,-latex] (axis cs:0.55,-0.4) -- (axis cs:0.55,0);
    	\nextgroupplot[enlargelimits=false,ylabel={},xtick={-0.5,0,0.5},xticklabels={$-\frac{\lambda}{2}$,$0$,$\frac{\lambda}{2}$},ytick={-0.5,0,0.5},yticklabels={,,},ylabel absolute,title={Ours}, ylabel style={yshift=-0.25cm}] 			
    	\addplot[] graphics[xmin=-0.75,ymin=-0.75,xmax=0.75,ymax=0.75] {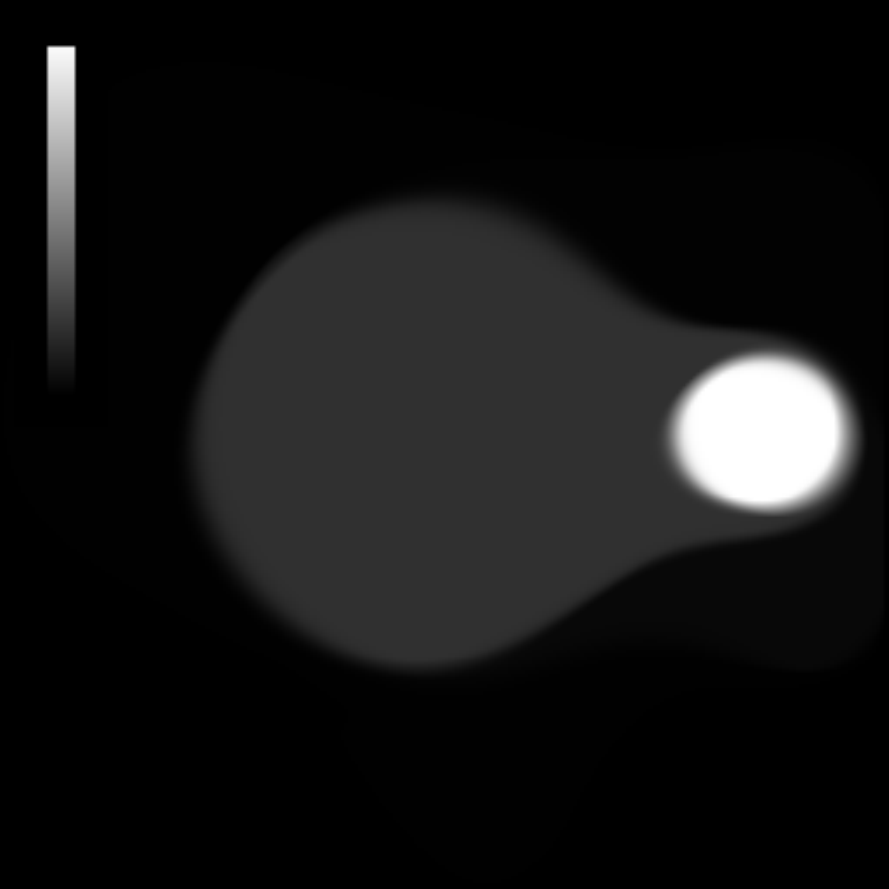};   	 	
    	\node[white] at (axis cs:-0.57,0.1) {{\scriptsize $1$}};
    	\node[white] at (axis cs:-0.52,0.67) {{\scriptsize $3.3$}};
    	\nextgroupplot[enlargelimits=false,xtick={-0.5,0,0.5},yticklabels={,,},xticklabels={$-\frac{\lambda}{2}$,$0$,$\frac{\lambda}{2}$},title={SEAGLE}] 	
    	\addplot[] graphics[xmin=-0.75,ymin=-0.75,xmax=0.75,ymax=0.75] {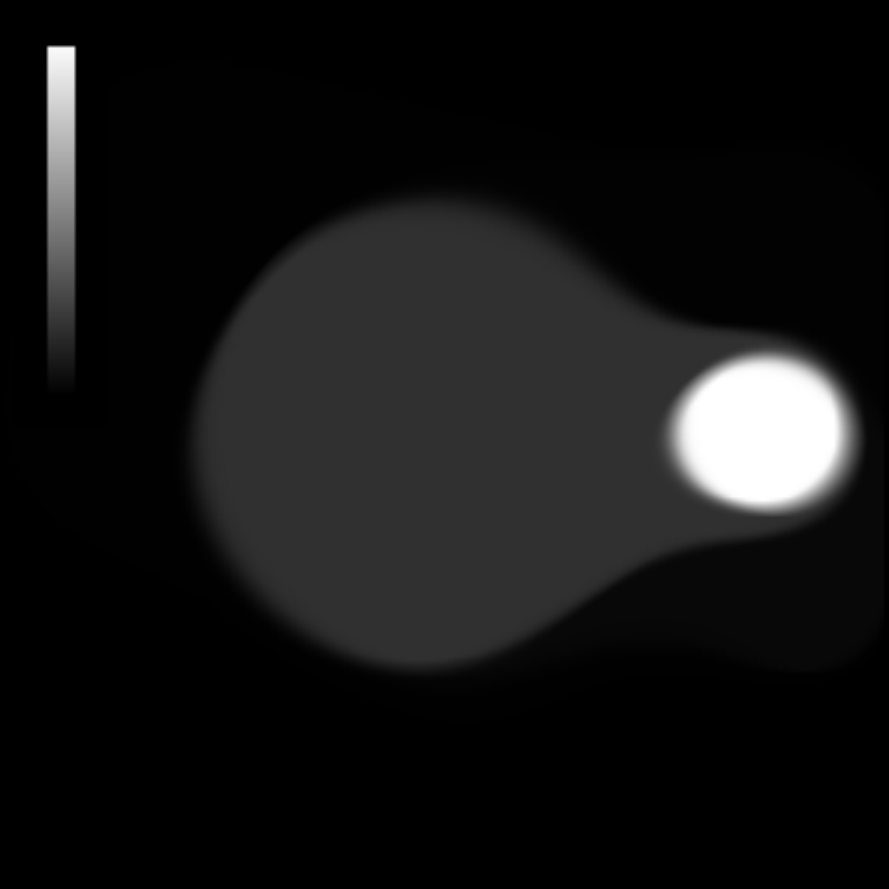};	   	
    	\node[white] at (axis cs:-0.57,0.1) {{\scriptsize $1$}};
    	\node[white] at (axis cs:-0.52,0.67) {{\scriptsize $3.3$}};
		\end{groupplot}
	\end{tikzpicture}
		\caption{\label{fig:ResultsReal} Reconstructions (permittivity) obtained by the proposed method and by SEAGLE for the \textit{FoamDielExt} target of the Institut Fresnel's database~\cite{Geffrin2005} with $\mu = 1.6 \cdot 10^{-2}$. The SNR values (computed from the experimentally measured permittivity of the ground truth) are $25.13$ dB (Ours) and $25.15$ dB (SEAGLE) while the computing times are respectively of $6$ min and $93$ min.}
	\end{figure}
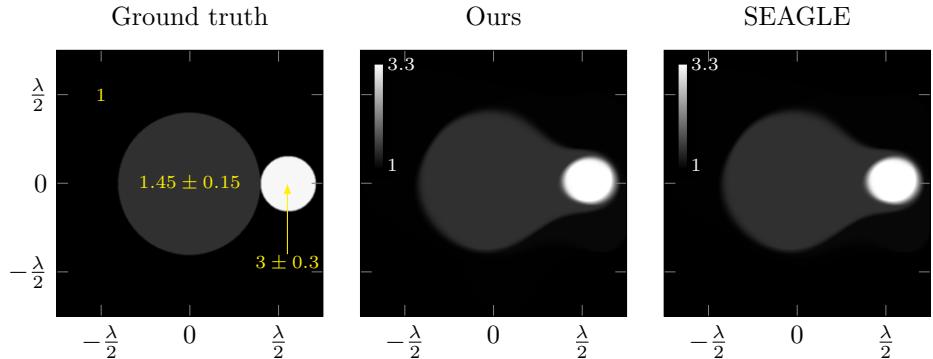

\section{Conclusion}

		We have presented a refinement of the SEAGLE algorithm that was recently proposed in \cite{Liu2016,Liu2017} and that has shown unprecedented reconstructions for difficult configurations. However, the current limitation of SEAGLE is  that its memory requirements increase excessively with the size of the problem, particularly in 3D. As an alternative, we have derived the explicit expression of the Jacobian matrix  $\Jd_{h_p}(\fd)$ of the nonlinear Lippman-Schiwnger model and shown that it can be computed in a direct analogy with the computation of the forward model. This approach allows us to drastically reduce the memory consumption and opens the door to 3D reconstruction using desktop computers.  Moreover, the proposed  method is quite flexible in the sense that it can cope with any iterative algorithm employed to compute either  the forward model or  $\Jd^H_{h_p}(\fd)$. For instance, the conjugate-gradient algorithm proved its efficiency for this task. It allows a significant decrease of the computational time with respect to SEAGLE. Finally, these improvements in terms of speed and memory come at no loss in quality.

\appendix 
\section{Proximity operator of $\mathcal{R}$} \label{apndx:Prox} 
 
 	In this appendix, we describe how we compute the proximity operator of the regularization term $\mathcal{R}$ in \eqref{eq:Regul} using the ADMM algorithm \cite{Boyd2011}.
The proximity operator is defined \cite{Moreau1962} as the solution of the optimization problem 
\begin{equation}\label{eq:ProxR}
	\mathrm{prox}_{\mu \mathcal{R}}(\vd) = \argmin{\fd \in \R^N}{ \left( \frac12\|\fd - \vd\|^2_2 + \mu\| \fd\|_{\mathrm{TV}}+ i_{\geqslant 0}(\fd) \right) }
\end{equation}
for $\mu >0$.  Let us start by reformulating  \eqref{eq:ProxR} as
\begin{align}
\mathrm{prox}_{\mu \mathcal{R}}(\vd) = \argmin{\fd \in \R^N}{& \left( \frac12\|\fd - \vd\|^2_2 + \mu\|\qd_1\|_{2,1}+ i_{\geqslant 0}(\qd_2) \right) }, \notag \\
s.t. \quad & \qd_1 = \nab \fd ,\notag \\
& \qd_2 = \fd,
\end{align}
which admits the  augmented-Lagrangian form
\begin{multline}\label{eq:lagrangian}
	\mathcal{L}(\fd,\qd_1,\qd_2,\wdd_1,\wdd_2) = \frac12 \|\fd - \vd \|^2_2 + \frac{\rho_1}{2}\left\| \nab \fd - \qd_1 + \frac{\wdd_1}{\rho_1} \right\|^2_2 \\ + \frac{\rho_2}{2}\left\| \fd - \qd_2 + \frac{\wdd_2}{\rho_2} \right\|^2_2 + \mu \|\qd_1\|_{2,1} + i_{\geqslant 0}(\qd_2), 
\end{multline}
where $\rho_1$ and $\rho_2$ are positive scalars, and where $\wdd_1 \in \R^{N \times D}$ and $\wdd_2 \in \R^N $ are the Lagrangian multipliers. Then, one can minimize \eqref{eq:lagrangian} using  ADMM. The iterates  are summarized in Algorithm~\ref{Algo:ADMMProx}.
	\begin{algorithm}[t]
		\caption{ADMM for solving \eqref{eq:ProxR}.}\label{Algo:ADMMProx}
	\begin{algorithmic}[1]
		\REQUIRE  $\fd^0 \in \R^N$, $\mu >0$, $\rho_1 >0$, $\rho_2 >0$
		\STATE $\mathbf{A}=\left( (1+\rho_2)\Id + \rho_1 \nab^T \nab\right)$
		\STATE $\qd_1^0 = \nab \fd^0$, $\qd_2^0 =  \fd^0$
		\STATE $\wdd_1 = \qd_1$, $\wdd_2 = \qd_2$
		\STATE $k=1$
		\WHILE{(not converged)}
			\STATE $\qd_1^{k+1} = \mathrm{prox}_{\frac{\mu}{\rho_1}\|\cdot\|_{2,1}  }\left( \nab \fd^k + \frac{\wdd_1^k}{\rho_1}\right) $
			\STATE $\qd_2^{k+1} = \mathrm{prox}_{i_{\geqslant 0}}\left( \fd^k + \frac{\wdd_2^k}{\rho_2}\right) $			
			\STATE$\fd^{k+1}= \mathbf{A} ^{-1} \left( \vd + \rho_1 \nab^T\left( \qd_1^{k+1} -\frac{\wdd_1^k}{\rho_1}   \right)       +\rho_2 \qd_2^{k+1}  -\wdd_2^k  \right)  $ \COMMENT{Fourier division}
			\STATE $\wdd_1^{k+1} = \wdd_1^k + \rho_1(\nab \fd^{k+1} - \qd_1^{k+1})$
			\STATE $\wdd_2^{k+1} = \wdd_2^k + \rho_2(\fd^{k+1} - \qd_2^{k+1})$
			\STATE $k=k+1$
 		\ENDWHILE
 	\end{algorithmic}
	\end{algorithm}
	
	For the sake of completeness, we provide in \eqref{eq:apendix1} and \eqref{eq:apendix2} the expressions
	\begin{align}
		\forall \qd \in  \R^N,  & \left[ \mathrm{prox}_{i_{\geqslant 0}}(\qd)\right]_n  = (\qd_n)_+, \label{eq:apendix1} \\
		\forall \qd \in  \R^{N \times D}, &  \left[ \mathrm{prox}_{\gamma \|\cdot\|_{2,1}  }(\qd)\right]_{n,d}   =  \qd_{n,d}	\left( 1 - \frac{\gamma}{\| \qd_{n,.} \|_2}  \right)_+, \label{eq:apendix2}
	\end{align}
	 of $\mathrm{prox}_{i_{\geqslant 0}}$ and $\mathrm{prox}_{\gamma \|\cdot\|_{2,1}  }$ where 
	\begin{equation}
	(x)_+ : = \max(x,0), x \in \R.
\end{equation}		


\section*{Acknowledgment}
This research was supported by the European Research Council (ERC) under the European Union's Horizon 2020 research and innovation programme, Grant Agreement no 692726 ``GlobalBioIm: Global integrative framework for computational bio-imaging.''


\end{document}